\documentclass[a4,11pt]{article}
 
\usepackage{microtype}

\usepackage{fullpage}

\AtBeginDocument{
  \providecommand\BibTeX{{
    \normalfont B\kern-0.5em{\scshape i\kern-0.25em b}\kern-0.8em\TeX}}}

\usepackage{microtype}

\usepackage{amsthm,amssymb,amsmath}
\usepackage{paralist}
\usepackage{enumerate}
\usepackage{comment}
\usepackage{url}

\usepackage{tikz}
\usetikzlibrary{automata,arrows,calc,decorations.pathmorphing,backgrounds,positioning,fit,petri}

\usepackage[active]{srcltx}
\usepackage{algorithm}
\usepackage{algpseudocode} 
\usepackage{graphicx}
\usepackage{color} 
\usepackage{latexsym}

\usepackage{color}
\usepackage{soul}

\usepackage{tcolorbox}
\tcbset{colback=white,colbacktitle=white!95!black,coltitle=black,left=0.5pt,right=1pt}

\newcommand{\namedref}[2]{\hyperref[#2]{#1~\ref*{#2}}}

\newcommand{\equalityref}[1]{\hyperref[#1]{Equality~\eqref{#1}}}
\newcommand{\inequalityref}[1]{\hyperref[#1]{Inequality~\eqref{#1}}}

\newenvironment{proof sketch}[1]{\noindent {\emph{Proof sketch of #1:}}}{\hfill \qed}

\newtheorem{theorem}{Theorem}
\newtheorem{claim}[theorem]{Claim}

\newtheorem{lemma}[theorem]{Lemma}
\newtheorem{definition}[theorem]{Definition}

\newtheorem{invariant}[theorem]{Invariant}

\newcommand{\Set}[1]{\left\{#1\right\}}
\renewcommand{\deg}{\textit{deg}}
\newcommand{\indeg}{\textit{deg}_{\textit{in}}}
\newcommand{\outdeg}{\textit{deg}_{\textit{out}}}
\newcommand{\degdist}{\Delta}
\newcommand{\indegdist}{\degdist_{\textit{in}}}
\newcommand{\pr}{\mathbf{Pr}}

\ifthenelse{\isundefined{\GenerateShortVersion}}
{
\def\LongVersion{}
\def\LongVersionEnd{}
\long\def\ShortVersion#1\ShortVersionEnd{}
}
{
\def\ShortVersion{}
\def\ShortVersionEnd{}
\long\def\LongVersion#1\LongVersionEnd{}
}

\newcommand{\poly}{\mbox{poly}}
\newcommand{\polylog}{\mbox{polylog}}

\newcommand{\parent}{\textit{parent}}

\newcommand{\potenf}{\text{$\Phi$}}

\newcommand{\nuc}{\texttt{next-child}}
\newcommand{\nuca}{\texttt{next-child-from}}
\newcommand{\nuct}{\texttt{next-child-flag}}
\newcommand{\up}{\texttt{parent}}
\newcommand{\naivenuc}{\texttt{na\"{\i}ve-\nuc}}
\newcommand{\childsccsr}{\texttt{successor}}
\newcommand{\childinsert}{\texttt{insert}}
\newcommand{\children}{\texttt{child}}
\newcommand{\toss}{\texttt{toss}}

\newcommand{\hextmin}{\texttt{heap-extract-min}}
\newcommand{\hinsert}{\texttt{heap-insert}}

\newcommand{\flyBA}{\texttt{O-t-F-BA}}
\newcommand{\optBA}{\texttt{Opt-BA}}

\newcommand{\nn}{\texttt{BA-next-neighbor}}
\newcommand{\fparent}{\texttt{BA-parent}}

\newcommand{\direct}{{\tt imm}}
\newcommand{\recursive}{{\tt rec}}
\newcommand{\constfalse}{{\tt false}}
\newcommand{\consttrue}{{\tt true}}
\newcommand{\whp}{w.h.p.}

\newcommand{\front}{\textit{front}}

\newcommand{\nil}{\mathsf{nil}}
\newcommand{\trank}{\textit{rank}}
\newcommand{\tselect}{\textit{select}}
\newcommand{\tinsert}{\textit{insert}}
\newcommand{\tdelete}{\textit{delete}}

\newcommand{\rhead}{\textit{head}}

\newlength{\mynodewidth}
\newlength{\mynodewidtht}
\tikzstyle{mynodestyle} = [draw, outer sep=2,inner sep=10,minimum size=20, text width=\mynodewidth,align=center,execute at begin node=\hskip0pt,rounded corners=3mm,top color=white,bottom color=black!20]
\tikzstyle{emptynodestyle} = [outer sep=2,inner sep=10,minimum size=20, text width=\mynodewidth]

\DeclareMathOperator*{\argmin}{arg\,min}

\begin{document}

\title{Sublinear Random Access Generators for Preferential Attachment Graphs
\thanks{A preliminary version of this work appeared in the Proceedings of ICALP 2017~\cite{EvenLMR17}.}}

 \author{Guy Even\thanks{
 Tel Aviv University Tel Aviv 6997801 Israel. Email: {\tt guy@eng.tau.ac.il}.}
 \and
Reut Levi\thanks{
The Interdisciplinary Center Herzliya (IDC). Email: {\tt reut.levi1@idc.ac.il}. This research was supported
by the Israel Science Foundation grant No. 1867/20.}
\and
Moti Medina\thanks{
 Faculty of Engineering, Bar-Ilan University, Ramat Gan, Israel. Email: {\tt moti.medina@biu.ac.il}. This research was supported by the Israel Science Foundation Grant No. 867/19.}
\and
Adi Ros\'{e}n\thanks{    
CNRS and Universit\'{e}  Paris Cit\'e, France. Email: {\tt adiro@irif.fr}. Research  supported in part by ANR project RDAM.   }
}
 
 \date{}

\maketitle

\begin{abstract}

We consider the problem of sampling from a distribution on graphs,  specifically
when the distribution is defined by an evolving graph model, and consider the
time, space and randomness complexities of such samplers.

 In the standard approach, the whole graph is chosen randomly according to
 the randomized evolving process,  stored in full, and then queries on the sampled graph are answered by simply accessing the stored graph. This may require prohibitive  amounts of
 time, space  and random bits, especially when only a small number of queries
are actually issued.
   Instead, we propose a setting where one  generates parts of the sampled  graph on-the-fly, in response to queries, and
   therefore   requires amounts of time, space, and random bits which are  a function of the actual number of
    queries.  Yet, the responses to the queries correspond to a graph sampled from the distribution in question.

  Within this framework we focus on two random graph models: the Barab{\'{a}}si-Albert Preferential
  Attachment model (BA-graphs) {\color{black}(Science, 286(5439):509--512)  (for the special case of out-degree $1$)} and the random recursive tree model {\color{black}(Theory of Probability and Mathematical Statistics, (51):1--28)}.
  We give on-the-fly generation algorithms for both models. With
  probability  $1-1/\poly(n)$, each and every query is answered in $\polylog(n)$ time, and the
  increase in space and the number of random bits consumed by any single query are both
   $\polylog(n)$, where $n$ denotes the number of vertices in the graph.

 Our work thus proposes a new approach for the access to huge graphs sampled from a given distribution,  and our  results show that, although the BA random graph model is defined by a sequential
  process, efficient  random access  to the graph's nodes is possible. In addition to the
  conceptual contribution, efficient on-the-fly generation of random graphs can serve as a
  tool for the efficient simulation of sublinear algorithms over large BA-graphs, and the efficient
  estimation of their performance on such graphs.

\end{abstract}

\section{Introduction}

Consider a Markov process in which a sequence $\{S_t\}_t$ of states, $S_t \in  \mathcal{S}$,  evolves over time $t \geq 1$.
Suppose there is a set $\mathcal{P}$ of predicates defined over the state space
$\mathcal{S}$. Namely, for every predicate $P\in\mathcal{P}$ and state $S\in\mathcal{S}$, the value
of $P(S)$ is well defined. A query is a pair $(P,t)$ and the answer to the query is
$P(S_t)$. In the general case, answering a query $(P,t)$ requires letting the Markov process
run for $t$ steps until $S_t$ is generated. In this paper we are interested in ways
to reduce the dependency, on $t$, of the computation time, the memory space, and the number
of used random bits, required to answer a query $(P,t)$. We propose an approach to achieve that in
the context of huge random graphs, samples according
to some given distribution.

We focus on the  case of generative models for random graphs, and in
particular, on the  Barab{\'{a}}si-Albert Preferential Attachment
model~\cite{BA99} with out-degree $1$ (which we call BA-graphs), on the equivalent linear evolving
copying model of Kumar et al.~\cite{KumarRRSTU00}, and on the random recursive tree model~\cite{survey_trees}.
The question we address is
whether one can design a randomized {\em on-the-fly}  graph generator
 that
answers  adjacency list queries of  BA-graphs (or random recursive trees), without having to generate the
complete graph.  Such a generator outputs answers to adjacency list queries as if it
first selected  the whole graph at random (according the appropriate distribution)
and then answered the queries based on the sampled graph.

We are interested in the following resources of a graph generator:
(1)~the number of
random bits consumed per query, (2)~the running time per query, and (3)~the increase
in memory space  per query.

Our main result is a randomized on-the-fly graph generator for BA-graphs over $n$ vertices that
answers adjacency list queries. The generated graph is sampled according to the distribution defined for BA-graphs  over $n$ vertices,
and the complexity upper bounds that we prove hold with probability
$1-1/\poly(n)$.  That is, with probability
$1-1/\poly(n)$  each and every query is answered in $\polylog(n)$ time, and the increase in space, and  the number
of  random bits consumed during that  query are $\polylog (n)$.  Our result refutes (definitely for
$\polylog(n)$ queries) the recent statement of Kolda et al.~\cite{KoldaPPS14} that:
``The majority of graph models add edges one at a time in a way that each random edge
influences the formation of future edges, making them inherently serial and therefore
unscalable.  The classic example is Preferential Attachment, but there are a variety
of related models...''

We remark that the entropy of the edges in BA-graphs is $\Theta(\log n)$ per edge in the
second half of the graph~\cite{sauerhoffentropy}. Hence it is not possible to consume
a sublogarithmic number of random bits per query in the worst case if one wants
to sample according to the  BA-graph distribution. Similarly, to
ensure consistency (i.e., answer the same query twice  in the same way) one must use $\Omega(\log n)$ space
per query.

From a conceptual point of view, the main ingredient of our result are techniques to ``invert'' the sequential
process where each new vertex randomly selects its ``parent'' in the graph among the previous vertices.
 Instead,  vertices randomly select
 their ``children'' among the ``future'' vertices, while maintaining the same
 probability distribution as if each child  picked
 ``in the future'' its parent. We apply these techniques in the related model of random recursive trees~\cite{survey_trees} (also used within
  the evolving copying model~\cite{KumarRRSTU00}), and use
 them as a building block for our main result for BA-graphs. 
{\color{black} We next define the various random graph models we refer to.}
\ShortVersion
 Due to space limitations, some of the proofs are omitted from this extended abstract.
\ShortVersionEnd

\subsection{Random Graph Models}\label{section.models}
Let $V_n \triangleq \{v_1,\ldots,v_n\}$.  Let $G=(V_n,E)$ denote a directed
graph on $n$ nodes.\footnote{Preferential attachment graphs are usually presented as undirected
  graphs. For convenience of discussion we orient each edge from the high index vertex to the low index vertex,
  but the graphs we consider remain undirected graphs.}
  We refer to the endpoints of a directed edge $(u,v)$ as the \emph{head} $v$ and the
\emph{tail} $u$.
Let $\deg(v_i,G)$ denote the \emph{degree} of the vertex $v_i$ in $G$ (both incoming
and outgoing edges). Similarly, let $\indeg(v_i,G)$ and $\outdeg(v_i,G)$ denote the
in-degree and out-degree, respectively, of the vertex $v_i$ in $G$.

\paragraph*{Preferential attachment~\cite{BA99}.}
We restrict our attention to the case  in which each vertex is connected to the previous
vertices by a single edge (i.e., $m=1$ in the terminology of~\cite{BA99}).
\LongVersion
\footnote{As mentioned-above,
while the process generates an undirected graph, for ease of discussion we consider
 each edge as directed from its higher-numbered adjacent node to its lower-numbered adjacent node.}
\LongVersionEnd
We thus denote the random process that generates a graph over $V_n$ according to
the preferential attachment model by $BA_n$.  The random process $BA_n$ generates a
sequence of $n$ directed edges $E_n\triangleq\{e_1,\ldots,e_n\}$, where the tail of
$e_i$ is $v_i$, for every $i\in \{1, \ldots , n\}$.  (We abuse notation and let $BA_n=(V_n,E_n)$
also denote the graph generated by the random process.)
We refer to the head of $e_i$ as the \emph{parent} of $v_i$.

The process $BA_n$ draws the edges sequentially starting with the self-loop $e_1
=(v_1,v_1)$.  Suppose we have selected $BA_{j-1}$, namely, we have drawn the edges
$e_1,\ldots, e_{j-1}$, for $j>1$.  The edge $e_j$ is drawn
such its head is node $v_i$ with probability  $\frac{\deg(v_i,BA_{j-1})}{2(j-1)}$.

Note that the out-degree of every vertex in (the directed graph representation of) $BA_n$ is exactly one, with only one
self-loop in $v_1$. Hence $BA_n$ (without the self-loop) is an in-tree rooted at $v_1$.

\paragraph{Evolving copying model~\cite{KumarRRSTU00}.}

Let $Z_n$ denote the evolving copying model with out-degree $d=1$ and copy factor
$\alpha=1/2$. As in the case of $BA_n$, the process $Z_n$ selects the edges
$E'_n=\{e'_1,\ldots,e'_n\}$ one-by-one starting with a self-loop $e'_1=(v_1,v_1)$. Given
the graph $Z_{n-1}=(V_n,E'_n)$, the next edge $e'_n$ emanates from $v_n$.  The head
of edge $e'_n$ is chosen as follows.
Let $b_n\in \Set{0,1}$ be an unbiased random bit.  Let $u(n)\in
\{1, \ldots , n-1\}$ be a uniformly distributed random variable
(the random variables $b_1,\ldots, b_n$ and $u(1),\ldots, u(n)$ are fully
independent.)  The head $v_i$ of $e'_n$ is determined as follows:
\ShortVersion
$\rhead(e'_n) \triangleq u(n)$, if $b_n=1$; and $\rhead(e'_n) \triangleq \rhead(e_{u(n)})$, if $b_n=0$.
\ShortVersionEnd
\LongVersion
\begin{align}\label{lb1}
\rhead(e'_n) &\triangleq
\begin{cases}
  u(n) & \text{if $b_n=1$}\\
\rhead(e'_{u(n)}) & \text{if $b_n=0$}
\end{cases}~.
\end{align}
\LongVersionEnd

\paragraph{Random recursive tree model~\cite{survey_trees}.}
If we eliminate from the evolving copying model the bits $b_i$ and  define  $\rhead(e'_n) \triangleq  u(n)$, we get a model
where each new node $n$ is connected to one of the previous nodes,  chosen uniformly
 at random. This is the extensively studied (random) recursive tree model~\cite{survey_trees}.

\LongVersion
\bigskip
\LongVersionEnd

{\color{black}
\subsection{Results}

Our main results are stated in the following theorems (which are informal versions of Theorem~\ref{th:comp_one_call_nn} and Lemma~\ref{le:complexity_pointers_tree}, respectively).

\begin{theorem}(Informal)
There exists an algorithm that provides query access to the adjacency-lists of a graph $G$ over $n$ nodes, where $n$ is a parameter and $G$ is drawn according to the random process $BA_n$.
With high probability, the complexities of executing each query are as follows.
\begin{enumerate}
\item The increase, during that query, of the space used  by our algorithm is $O(\log^3 n)$.
\item The number of  random bits used during that query is $O(\log^5 n)$.
\item The time complexity of that query is $O(\log^6 n)$.
\end{enumerate}
\end{theorem}

\begin{theorem}(Informal)
There exists an algorithm that provides query access to the adjacency-lists of a graph $G$ over $n$ nodes, where $n$ is a parameter and $G$ is drawn according to the random process $Z_n$.
With high probability, the complexities of executing each query are as follows.
\begin{enumerate}
\item The increase, during that query, of the space used  by our algorithm is $O(\log^2 n)$.
\item The number of  random bits used during that query is $O(\log^4 n)$.
\item The time complexity of that query is $O(\log^5 n)$.
\end{enumerate}
 \end{theorem}

}

\LongVersion
\subsection{Related work}
\LongVersionEnd
\ShortVersion
\vspace{-0.4cm}
\subparagraph*{Related work.}
\ShortVersionEnd
A linear time randomized algorithm for efficiently generating BA-graphs is given in
Batagelj and Brandes~\cite{batagelj2005efficient}. See also Kumar et
al.~\cite{KumarRRSTU00} and Nobari et al.~\cite{nobari2011fast}.
A parallel algorithm
is given in Alam et al.~\cite{AlamKM13}. See also Yoo and
Henderson~\cite{yoo2010parallel}. An external memory algorithm was presented by Meyer
and Peneschuck~\cite{extmem2016}.
{\color{black}

Efficient generation of other graph models was also studied.
Miller and Hagberg~\cite{miller2011efficient} introduced a randomized algorithm that generates a graph with a given sequence of
expected degrees (also called the Chung and Lu model) with expected running time of $O(n+m)$, where $n$ is the number of vertices
and $m$ is the number of edges of the generated graph.
Additional efficient random graph generation algorithms for other graph
models (e.g., Kronecker and the Stochastic block model) are provided in
 Ramani, Eikmeier, and Gleich~\cite{ramani2019coin}.}
\ShortVersion
Generating huge random objects while using ``small'' amounts of randomness was studied by
Goldreich, Goldwasser and Nussboim~\cite{GGN10}. Mansour et al.~\cite{MansourRVX12}
consider local generation of bipartite graphs in the context of  local simulation of Balls into Bins online algorithms.
\ShortVersionEnd
\LongVersion

Goldreich, Goldwasser and Nussboim initiate the study of the generation of huge random objects~\cite{GGN10}
while using a ``small''  amount of randomness. They  provide an efficient {\color{black} stateless} query access to an object modeled as a function, when
the object has a predetermined  property, for example graphs which are connected.  They guarantee that these
objects are indistinguishable from random objects that have the  same property.
This refers to the setting where the size of the object is exponential in the number of queries to the function modeling
the object.
{\color{black}In a followup paper by Bogdanov and Wee~\cite{bogdanov2004stateful} stateful implementations of huge random objects were considered. They showed how to
generate in an ``on the fly'' fashion a random  Boolean function that supports XOR queries over sub-cubes of the function's domain hypercube.}
We note that our {\color{black} stateful} generator provides access to graphs which are random BA-graphs and not just indistinguishable from random BA-graphs.

Mansour, Rubinstein, Vardi and Xie~\cite{MansourRVX12} consider local generation of bipartite graphs
for local simulation of Balls into Bins online algorithms.
They assume that the balls arrive one by one and that each ball picks $d$ bins independently, and is then assigned to
one of them. The local simulation of the  algorithm locally generates a bipartite graph.
Mansour et al. show that with high probability one needs to inspect only a small portion of the the bipartite
graph in order to run the simulation and hence a random seed of logarithmic size
is sufficient.

Our work has inspired subsequent work in the setting that we propose here, i.e., on-the-fly local generation of
graphs according to a given distribution. In fact, subsequent to the initial publication of the present
work~\cite{EvenLMR17}, Biswas, Rubinfeld, and Yodpinyanee~\cite{BisRY17}
 have devised local graph generators for, most notably, the Erd\"os-R\'enyi model (with
 next-neighbor, and other,  queries).

\LongVersionEnd
\LongVersion
\subsection{Applications}
\LongVersionEnd
\ShortVersion
\vspace{-0.4cm}
\subparagraph*{Applications.}
\ShortVersionEnd

One reason for generating large BA-graphs is to simulate algorithms over them, or to experimentally estimate
some of their properties (cf.~\cite{DraKMM16}).  Such
algorithms often access only small portions of the graphs. In such instances, it is
wasteful to generate the whole graph. An interesting example is sublinear
approximation algorithms~\cite{dana2012,yoshida2009improved,onak2008,onakthesis}
which probe a constant number of neighbors.
\LongVersion
\footnote{Strictly speaking, sublinear
  approximation algorithms apply to constant degree graphs and BA-graphs are not
  constant degree.  However, thanks to the power-law distribution of BA-graphs, one
  can ``omit'' high degree vertices and maintain the approximation. See
  also~\cite{ReingoldV14}.}
\LongVersionEnd
In addition, local computation algorithms probe a small
number of neighbors to provide answers to optimization problems such as maximal
independent sets and approximate maximum matchings~\cite{EvenMR14, EvenMResa14,
  ReingoldV14, RubinfeldTVX11, AlonRVX12, MansourRVX12, MansourV13, LeviMRRS15,
  LeviRR14, LeviRY15}. Support of adjacency list queries is especially useful for
simulating (partial) DFS and BFS over graphs.

\LongVersion
\subsection{Techniques}
 The main difficulty in providing the on-the-fly generator is in ``inverting'' the random choices of the BA process. That is,
 we need to be able to randomly choose the next ``child'' of a given node $x$, although it  will only ``arrive in the future'' and its
 choice of a parent in the BA-graph will depend on what will have happened until it arrives (i.e., on the node degrees in the BA-graph when
 that node arrives). One possibility to do so is to maintain, for any future node which does not yet have a parent,  how many potential
 parents it still has,  and then go sequentially over the future nodes and randomly decide if its parent will indeed be $x$. This is too
 costly  because  (1) we will need to go sequentially over the nodes, and (2)  it may be too costly in computation time to
 calculate what is the probability that the parent of a node $y$ that does  not have yet a parent, will be node $x$
 (given the random choices already done in response to previous
 queries).

To overcome this difficulty we define  for any node, even if it has already a parent, its probability to be a {\em candidate} to be a child of $x$.
We show how these probabilities can be calculated efficiently given the previous choices taken in response to previous queries, and show how, based on these probabilities,  we can define an efficient
process to chose the next {\em candidate}. The candidate  node may however
already have a parent, and thus cannot be a child of $x$. If this is the case we repeat the process  and  choose another candidate,
until we chose an eligible candidate which then is chosen to be the actual next child of $x$. We show that with high probability this process terminates
quickly and finds an eligible candidate, so that with high probability we have an efficient process to find ``into the future'' the next child of
 $x$. This is done while sampling exactly according to the distribution defined by the BA-graphs process.

In addition to the above technique, which is arguably the crux of our result, we use a number of data structures, based on known
constructions, to be able to run the on-the-fly generator with polylogarithmic time and space complexities.
  In the sequel we give, in addition to the formal definitions of the algorithms, some supplementary intuitive explanations into our
 techniques.
\LongVersionEnd

\section{Preliminaries}\label{sec:preliminaries}

\LongVersion
 The
\emph{normalized degree distribution} of $G$ is a vector $\degdist(G)$ with $n$
coordinates, one for each vertex in $G$.  The coordinate corresponding to $v_i$ is
defined by
\begin{align*}
  \degdist(G)_i & \triangleq  \frac{\deg(v_i,G)}{2\cdot |E|}~.
\end{align*}
Note that $\sum_{i=1}^n \degdist(G)_i=1$.

We also define the in-degree distribution $\indegdist(G)$ by
\begin{align*}
  \indegdist(G)_i & \triangleq  \frac{\indeg(v_i,G)}{|E|}~.
\end{align*}
\LongVersionEnd

In the sequel, when we say that an event occurs {\em with high probability} (or {\em w.h.p}) we mean that
it occurs with probability at least $1-\frac{1}{n^c}$, for some constant $c>0$.
{\color{black} We state the complexities of our algorithm (and our subroutines) with the guarantee of high probability.
Therefore, there is some negligible probability that our algorithm will require more resources~\footnote{\color{black}
We note that in any case the resources consumed by the algorithm can be bounded by the resources consumed by the sequential standard algorithm since one can easily abort and implement a sequential algorithm in the unlikely event that the on-the fly generator consumes too much resources.}.
}

\ShortVersion
For ease of presentation,  we  use in the algorithms arrays of size $n$. However,
in order to  keep the space complexity low,  we implement these
arrays by means of balanced search trees, with  keys in $\{1, \ldots , n\}$.  Thus, the space used by the
 ``arrays'' is  the number of keys stored. The time complexities that we give are therefore
 to be multiplied by a factor of $O(\log n)$.
\ShortVersionEnd

\LongVersion
For ease of presentation,  define the algorithm making use of {\em arrays} of size $n$. However,
in order to give the desired upper bounds on the space complexity, we implement these
arrays by means of balanced search trees, where the keys are in $\{1, \ldots , n\}$. To access
item $i$ in the virtual array,   key $i$ is searched in the tree and the value in that node is returned;
if the key is not found, then $\nil$ is returned. Thus, the space used by the
 virtual arrays is  the number of keys stored, and the time complexity of our
 algorithms is multiplied by a factor of $O(\log n)$ compared to the time complexity
 that it would have
  with a standard random-access implementation of the arrays. When we state upper bounds
   on time, we take into account these $O(\log n)$ factors.
  As common, we analyze the space complexity in terms of words of size $O(\log n)$.
\LongVersionEnd

\LongVersion
\section{Queries}
\LongVersionEnd
\ShortVersion
\section{Queries and On-the-Fly Generators}
\ShortVersionEnd
\label{se:queries}
\sloppy
Consider an undirected  graph $G=(V_n,E)$, where $V_n=\{v_1,\ldots,v_n\}$.  Slightly  abusing notation,
we sometimes consider and denote node $v_i$ as the integer number $i$ and so we have a natural order
 on the nodes. The access to the graph
 is done by means of a user-query $\nn: \{1, \ldots , n\} \rightarrow \{1, \ldots , n+1\}$, {\color{black} which outputs entries of the adjacency list in increasing order} (where $n+1$ denotes
 ``no additional neighbor'').
  {\color{black}
 For example, consider the node $3$ (namely the node that arrived third) whose parent is $1$ and its children are $7$, $10$ and $50$.
When the user executes the query $\nn(3)$ for the first time, $1$ is returned. The next time the query $\nn(3)$ is executed $7$ is returned, and then $10$ and $50$. After that, $\nn(3)$ always returns $n+1$ since $3$ does not have any additional neighbors.
 More formally}, we number the queries according to the order they are issued, and call
 this number the {\em time} of the query. Let $q(t)$ be the node on which the query at
 time $t$ was issued, i.e, at time $t$ the query  $\nn(q(t))$ is issued by the user.
 For each node $j\in V$ and any time $t$, let $last_t(j)$ be the largest numbered node
which was previously returned as the value of $\nn(j)$, or $0$ if no such
query was issued before time $t$. That is,
\ShortVersion
$ last_t(v) = \min\{0, \min_{t'<t}\{\nn(q({t'})) |  q({t'})=v\}$.
\ShortVersionEnd
\LongVersion
$$ last_t(j) = \max\{0, \max_{t'<t}\{\nn(q({t'})) |  q({t'})=j\}\}~.$$
\LongVersionEnd
At time $t$ the query $\nn(j)$ returns $\argmin_ {i> last_t(j)} \{(i,j) \in E\}$, or
$n+1$ if no such $i$ exists.
When the implementation of the query has access to a data structure holding the whole
of $E$, then the implementation
of \nn\ is straightforward just by accessing this data structure.
\LongVersion
Figure~\ref{fig:batch_query} illustrates
a ``traditional'' randomized graph generation algorithm that generates the whole graph, stores it, and then can
 answers queries by accessing the data structure that encodes the whole generated graph.
\LongVersionEnd

\LongVersion
\section{On-the-fly Graph Generators}
\LongVersionEnd
\sloppy
An on-the-fly graph generator is an algorithm that gives access to a graph by means
of the \nn\ query defined above, but itself does not have access to
a data structure that encodes the whole graph. Instead, in response to the queries
issued by the user, the generator modifies its internal data structure (a.k.a state),
which is initially some empty (constant) state.
The generator must ensure however that its answers are consistent with some graph $G$.
An on-the-fly graph generator for a given distribution on a family of graphs (such as the family of
Preferential Attachment graphs on $n$ nodes) must in addition ensure that it
samples the graphs according to the required distribution. That is,
 its answers to a sequence of queries must be  {\em  distributed identically}
 to those returned when a graph was first sampled (according to the desired distribution), stored,
  and then accessed (See Definition~\ref{def:answers}  and Theorem~\ref{th:equal_distribution}).
\LongVersion
 Figure~\ref{fig:otf_query} illustrates
an on-the-fly graph generation algorithm as the one we build in the present paper.
\LongVersionEnd

\LongVersion
\begin{figure}
  \centering
\begin{tikzpicture}[->,>=stealth,shorten >=1pt,auto,node distance=3cm,on grid,semithick]
\node[mynodestyle] (user) {user};
\node[mynodestyle] (mem) [below =of user] {memory (graph)};
\node[mynodestyle] (gen) [left=of mem] {batch Graph Generator};
\node[mynodestyle] (rnd) [left=of gen] {random bits};

\path 
	(user)  edge [bend left] node [right] {query}   (mem)
	(mem) edge [bend left] node [left]    {answer} (user)
	(rnd) edge 					 node [left]    {}           (gen)
	(gen) edge 					 node [left]    {}           (mem)
;

\end{tikzpicture}
  \caption{A ``traditional'' sequential random graph generator}
  \label{fig:batch_query}
\end{figure}
\LongVersionEnd

\LongVersion
\begin{figure}
  \centering
\begin{tikzpicture}[->,>=stealth,shorten >=1pt,auto,node distance=3cm,on grid,semithick]
\node[mynodestyle] (user) {user};
\node[mynodestyle] (gen) [below=of user] {on-the-fly Graph Generator};
\node[mynodestyle] (rnd) [below left=of gen] {random bits};
\node[mynodestyle] (mem) [below right=of gen] {memory (state)};

\path 
	(user)  edge [bend left] node [right] {query}   (gen)
	(gen) edge [bend left] node [left]    {answer} (user)
	(rnd) edge 					 node [left]    {}           (gen)
	(gen) edge [bend left] node [left]     {}           (mem)
	(mem) edge [bend left] node [right] {}       (gen)
;

\end{tikzpicture}
  \caption{An on-the-fly random graph generator}
  \label{fig:otf_query}
\end{figure}
\LongVersionEnd

We now relate the various models {\color{black} defined in Section~\ref{section.models}. The following relation is crucial for our implementation of the on-the-fly generator for BA-graphs. }
\ShortVersion
Proof omitted from this extended abstract.
\ShortVersionEnd
\begin{claim}[\cite{AlamKM13}]
  The random graphs $BA_n$ and $Z_n$ are identically distributed.
\end{claim}
\LongVersion
\begin{proof}
  The proof is by induction on $n$. The basis ($n=1$)  is trivial. To prove the induction step,
  assume that $BA_{n-1}$ and $Z_{n-1}$ are identically
  distributed. We need to prove that the next edges $e_n$ and
  $e'_n$ in the two processes are also identically distributed,
  given a graph $G$ as the realization of $BA_{n-1}$ and $Z_{n-1}$, respectively.

  The head of $e_n$ is chosen according to the degree distribution
  $\degdist(BA_{n-1})=\degdist(G) $. Since the out-degree of every vertex is one,
  \begin{align*}
    \frac{\deg(v_i,BA_{n-1})}{2(n-1)} & =
    \frac 12 \cdot \left(
\frac{1}{n-1} + \frac{\indeg(v_i,BA_{n-1})}{n-1}
\right) .
  \end{align*}
  Thus, an equivalent way of choosing the head of $e_n$ is as follows: (1)~ with
  probability $1/2$, choose a random vertex uniformly (this corresponds to the $\frac
  12 \cdot \frac 1{n-1}$ term), and (2)~with probability $1/2$ toss a
  $\indegdist(BA_{n-1})$-dice (this corresponds to the $\frac 12 \cdot
  \frac{\indeg(v_i,BA_{n-1})}{n-1}$ term).

Hence, case (1) above  corresponds to the case when $b_n=1$, in the
process of $Z_n$.
 To complete the proof, we observe that, conditioned on the event that
  $b_n=0$, the choice of the head of $e'_n$ in $Z_n$ can be defined as
choosing according to the in-degree distribution of the nodes in $Z_{n-1}=G$:
 indeed, choosing according
  to the  in-degree distribution $\indegdist(G)$ is identical to choosing a uniformly
  distributed random edge in $G$ and then taking its head. But, since  the out-degrees of
 all the vertices in $V_{n-1}$  are all the same (and equal one), this is equivalent to choosing
 a uniformly distributed random node in $V_{n-1}$.
  \end{proof}
\LongVersionEnd

We use the following claim  in the sequel.
\begin{claim}[cf.~\cite{goh2002limit}, Thm. 1 and Thm. 6.32~\cite{Drmota2009}]
\label{cl:low_degree_and_height}
Let $T$ be a rooted directed tree on $n$ nodes denoted $1,\ldots,n$, and where node $1$ is the root of the tree.  If the head of the edge emanating
 from node $j >1$ is uniformly
distributed among the nodes in $[1,j-1]$, then, with high probability, the following two properties hold:
\ShortVersion
\begin{inparaenum}[(1)]
\ShortVersionEnd
\LongVersion
\begin{enumerate}
\LongVersionEnd
\item The maximum in-degree of a node in the tree is $O(\log n)$.
\item The height of the tree is $O(\log n)$.
\LongVersion
\end{enumerate}
\LongVersionEnd
\ShortVersion
\end{inparaenum}
\ShortVersionEnd
\end{claim}

\LongVersion
Note that the claim still holds if we add to the tree a self loop on node $1$.
\LongVersionEnd

\section{The Pointers Tree}\label{sec:pointers}

We now consider a graph inspired by the  random recursive tree model~\cite{survey_trees} and the evolving copying
 model~\cite{KumarRRSTU00}.
 Each vertex $i$ has a variable $u(i)$ that is
uniformly distributed over $[1,i-1]$, and can be viewed as a
directed edge (or pointer) from $i$ to $u(i)$.
We denote this random rooted directed in-tree {\color{black} (namely, a tree such that all its edges point towards the root)} by $UT$.
 Let $u^{-1}(j)$ denote
the set $\{i : u(i)=j\}$. We refer to the set $u^{-1}(i)$ as the \emph{u-children} of
$i$ and to $u(i)$ as the \textit{u-parent} of $i$.

In conjunction with each pointer, we keep a flag indicating
whether this pointer is to be used as an {\em \direct} (immediate) pointer, that is, whether it points to the \rhead, or as a {\em \recursive} (recursive) pointer (as defined in Equation~\ref{lb1}). We thus
 use  the directed pointer tree to represent a graph in the evolving copying model (which is equivalent,
when the flag of each pointer is equally distributed between \recursive\ and \direct,  to the
BA model).
See Figure~\ref{fig:rec_tree2} for an illustrative example. 

\begin{figure}
  \begin{center}
    \includegraphics[width=.6\linewidth]{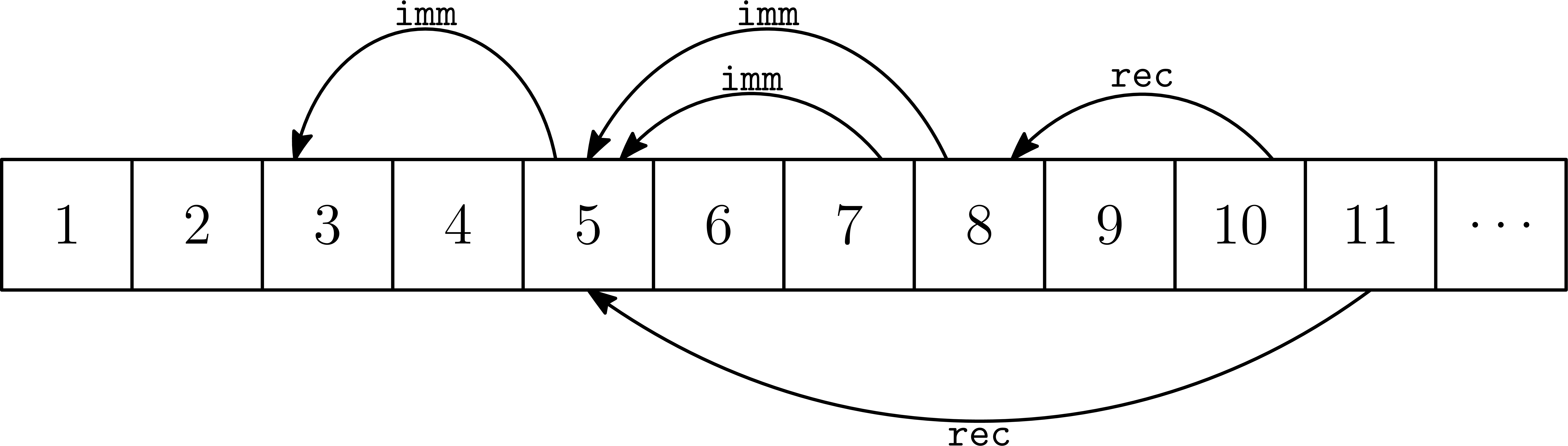}
  \end{center}
  \caption{\color{black} \textbf{The pointers tree.} The parent of node $5$ is node $3$ and the children of node $5$ are nodes $7$, $8$, and $10$ (since $10$ has a recursive pointer to node $8$). Node $11$ is a child of node $3$ since it has a recursive pointer to node $5$ (in particular, it is not a child of node $5$).}
  \label{fig:rec_tree2}
\end{figure}

In this section we consider the subtask of giving access to a random $UT$, together with the
flags of each pointer. Ignoring the flags, this section thus gives an on-the-fly random access generator
 for the extensively studied model of random recursive trees (cf.~\cite{survey_trees}).

We define the following queries.
\begin{itemize}
\item $(i,flag) \leftarrow \up(j)$:  $i$ is the parent of $j$ in the tree, and $flag$ is the
associated flag.
\item $i \leftarrow \nuct(j,k,flag)$, where $k \geq j$:  $i$ is the least numbered node $i>k$ such
that the parent of $i$ is $j$ and the flag of that pointer is  of {\color{black} type} $flag$.
If no
such node exists then $i$ is $n+1$.
\end{itemize}
{\color{black}
Given a query $\nuct(j,k,\cdot)$, we assume that $k$ is bounded above by the largest value returned by $\nuct(j,\cdot,\cdot)$ thus-far (and $j$ if it is the first time $\nuct(j,\cdot,\cdot)$ is executed).
Clearly, even under this assumption the neighbors of every node can be revealed one by one by repeated calls to $\nuct$ (by setting $k$ to be the returned value of the former execution of $\nuct$).
In this case the output of the queries is the entries of the adjacency list, in increasing order (where the entries with the wrong flag are filtered).
}

The ``ideal'' way to implement $\nuct$ is to go over all $n$ nodes, and for each node $j$ (1) uniformly at random
choose its parent  in $[1,j-1]$, (2) uniformly at random chose the associated flag
in $\{\direct,\recursive\}$. Then store  the  pointers and  flags, and answer the queries by
accessing this data structure.

\ShortVersion
In this section we give an {\em on-the-fly} generator that answers the above queries. We start with some notations.
We say that $j$ is \emph{exposed} if $u(j)\neq \mathsf{nil}$ (initially all pointers $u(j)$ are set to $\mathsf{nil}$).  We denote the set of
all exposed vertices by $F$.  We say that $j$ is \emph{directly} exposed if $u(j)$
was set during a call to  $\nuct(i,\cdot,\cdot)$.  We say that $j$
is \emph{indirectly} exposed if $u(j)$ was determined during a call to
$\up(j)$.
\ShortVersionEnd
\LongVersion
 In this section we give an {\em on-the-fly} generator that answers the above
queries.
{\color{black} As explained in more detail in the sequel, randomly selecting the parent of a node, to answer a $\up$ query, is a rather easy task (even if the generator already has a state). The challenge is in randomly selecting the children of a node which corresponds to \nuc\ queries. In this case we need to randomly select (according to the appropriate distribution) the first child of $j$ between \{j+1, \ldots, n\} and then the second child and so on.}
In what follows, we start with a na\"ive, non-efficient implementation that illustrates the task to be done. Then we give our efficient implementation.

\subsection{Notations}
We say that $j$ is \emph{exposed} if $u(j)\neq \mathsf{nil}$ (initially all pointers $u(j)$ are set to $\mathsf{nil}$).
  We denote the set of
all exposed vertices by $F$.
\LongVersionEnd
As a result of answering and processing \nuct\ and \up\ queries, the on-the-fly generator
commits to various decisions (e.g., prefixes of adjacency lists).
These commitments include
edges but also non-edges (i.e., vertices that can no longer serve as $u(j)$ for a certain $j$).
{\color{black}Therefore each query may change the state of the generator.
Note that the answers of the generator to queries depend on its state, thus,
the queries $\up$ and $\nuct$ are also a function of this state (which is not given as a parameter).

\subsubsection{The front of a node}
For every node $i\in \{1, \ldots , n-1\}$, the generator saves the {\em front} of $i$, which is roughly speaking, a value $k > i$ for which it holds that (1) $u(k)=i$;  and (2) the generator decided for every node $j \in [i+1,k-1]$ whether $u(j)=i$ or not.
Formally, at any given time, $\front(i)$ is a pointer to a node in $[i+1, n+1]$ which has the following properties.
 \begin{enumerate}
\item Initially $\front(i) =  \nil$. The first time $\front(i) \neq  \nil$ is after the first $\nuct(i,\cdot,\cdot)$ query is issued.
 \item If $\front(i) = \nil$, then for any node in $j\in [i+1, n]$ for which $u(j) = \nil$ it is possible that $u(j)$ will be set to $i$ in the future.
 \item If $\front(i) = k$ then
 \begin{enumerate}
  \item for any node in $j\in [i+1, k-1]$ for which $u(j) = \nil$ it is {\em not} possible that $u(j)$ will be set to $i$ in the future
 \item for any node in $j\in [k+1, n]$ for which $u(j) = \nil$ it is possible that $u(j)$ will be set to $i$ in the future
 \end{enumerate}
 \end{enumerate}

\subsubsection{The set of potential parents}

We denote the set of vertices that can become $u$-parents of $j$ at any given time $t$, by $\Phi(j)$ and their number by $\varphi(j)$ (for brevity, we omit $t$ from the notation). The formal definition is as follows.
\begin{definition}\label{def:phi}
At a given time $t$, and for any node $j$, let $\Phi(j)$ and $\varphi(j)$ be defined as follows:
\ShortVersion
$\Phi(j) \triangleq\{ i ~| ~i < j~ \mbox{\tt and}~(\front(i) <j~\mbox{\tt or} ~\front(i)=\nil)\}$, and $\varphi(j)=|\Phi(j)|$.
\ShortVersionEnd
\LongVersion
$$
\Phi(j) \triangleq\{ i ~| ~i < j~ \mbox{\tt and}~(\front(i) <j~\mbox{\tt or} ~\front(i)=\nil)\},~~ \mbox{and}~~\varphi(j)=|\Phi(j)|~.
$$
\LongVersionEnd
\end{definition}

We note that from technical reasons that will become clear below, according to the definition, $\Phi(j)$ is not necessarily empty even if $u(j)$ is already determined.
Moreover, counter intuitively, it might be the case that $u(j) \notin \Phi(j)$.
}

\LongVersion
\subsection{A na\"{\i}ve implementation of $\nuc$}

\begin{figure}

\begin{tcolorbox}[title = {\footnotesize \naivenuc}]
\begin{algorithmic}[1]
{\footnotesize
\Procedure{$\naivenuc$}{$j,k$}
	\State $x\gets k+1$.
	\While{$x\leq n$}
		\If {$u(x) = j$} \textbf{return} $(x)$
		\Else
		    \If {$u(x)=\nil$}
			\State Flip a random bit $c(x)$ such that $\pr[c(x)=1] = 1/\varphi(x)$.
			\If {$c(x)=1$}  \State \textbf{return} $(x)$
			\EndIf
		    \EndIf	
		\EndIf
		\State $x\gets x+1$
	\EndWhile
	\State \textbf{return} $(n+1)$
\EndProcedure
}
\end{algorithmic}
\end{tcolorbox}
  \caption{pseudo code of \naivenuc
  \label{alg:naive_nuc}
  }
\end{figure}

\begin{figure}
  \begin{center}
    \includegraphics[width=.6\linewidth]{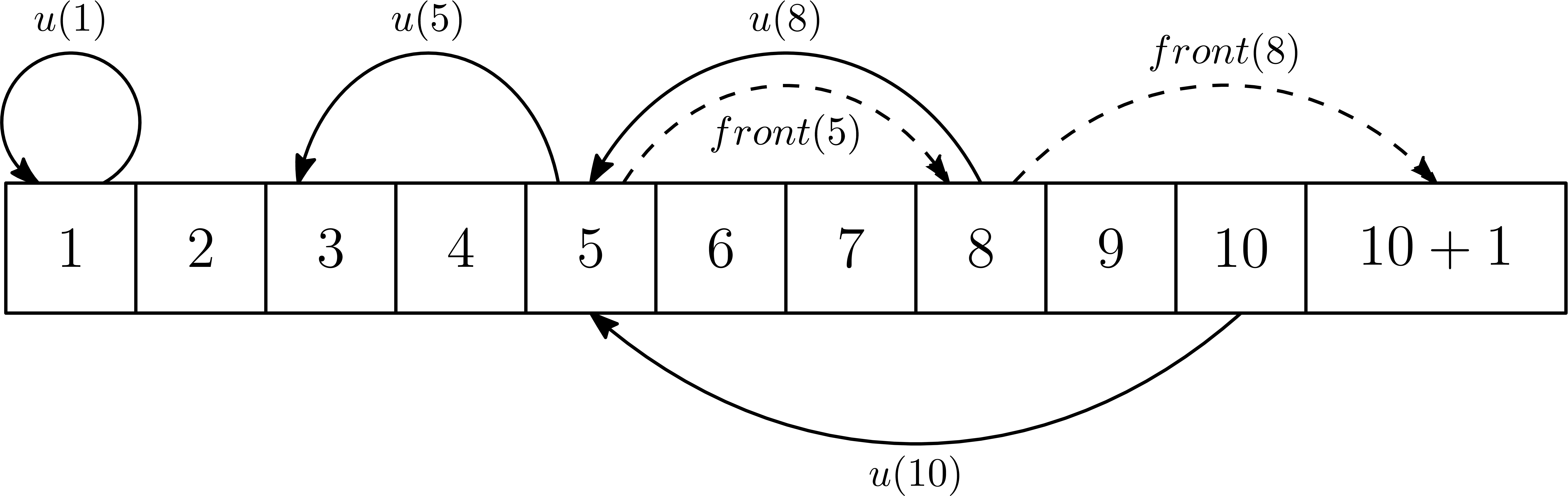}
  \end{center}
  \caption{\color{black} \textbf{The state of the generator}. The state of the generator after performing a $\nuc(5)$ and a $\up(10)$ queries. After the first query, which returns $8$, the parent of node $5$ is determined (in order to keep Invariant~\ref{assume:father first}) and set to be node $3$ and $\front(5)$ and the first child of node $5$ are set to be $8$. Consequently, in order to keep Invariant~\ref{inv:inv1} the next child of node $8$ is also discovered and is set to $10 +1$ (it has no children). After the second query, the parent of node $10$ is set to be node $5$. However, this does not change $\front(5)$ and so, potentially, node $9$ could also pick node $5$ as its parent (or any other node in $\{1, \ldots, 7\}$).}
  \label{fig:rec_tree}
\end{figure}

We give a
 na\"{\i}ve implementation of a \nuc\ query, with  time complexity
$O(n)$, with the  purpose of illustrating the main properties of this query  and in order to contrast it with the more efficient implementation later. We do so in a simpler manner without looking into the ``flag''.
 \LongVersion
The na\"{\i}ve implementation of  \nuc\  is listed in Figure~\ref{alg:naive_nuc}.
This implementation,
and that of \up, share an array of pointers $u$,  both updating it.
  A query $\nuc(i,k)$ is processed by scanning the vertices one-by-one starting from
$k+1$. If $u(x)=i$, then $x$ is the next child.
If $u(x)$ is $\nil$, then
a coin $c(x)$ is flipped and $u(x)=i$ is set when $c(x)$ comes out $1$; the probability that $c(x)$ is $1$ is
 $1/\varphi(x)$. If $c(x)=0$,  we proceed to the next
vertex. The loop ends when some $c(x)$ is $1$  or all vertices have been exhausted. In the latter case
 the query returns $n+1$.

{\color{black}
The correctness of  \naivenuc, i.e., the fact that the graph is generated according to the required probability distribution, is based on the observation that
given that $u(x)$ has not been determined yet,
all the vertices in $\potenf(x)$ are equally likely to serve as $u(x)$.  Note that the  description above   does not
explain how $\varphi(x)$ is computed, as explained next, this is one of the main challenges in designing our on-the-fly generator.  }
\LongVersionEnd

\subsection{The challenge in obtaining an efficient implementation of \nuc}\label{ef.sec}

\ShortVersion
We first shortly discuss the challenges on the way to an efficient implementation of
\nuc.
Observe that before the  first $\nuc(j)$ query, for a given $j$, is issued, the probability for
any $x >j$ to be a $u$-child of $j$  is $1/\varphi(x)$, because all nodes $x'<x$ can still be
the $u$-parent of $x$. But once $\nuc(\cdot)$ queries are  issued, this  may no  longer
be the case. For example, if $x> \front(j')$, then, even if the $u$-parent of $x$ is not
 yet determined, $j'$ is no longer an option to be the $u$-parent of $x$. This renders the calculation
 of $\pr[u(x)=j]$  more complicated and more computation-time consuming,
 which renders the process of selecting the next child of a node $j$ non-efficient.
 In the rest of this section we show how to overcome these difficulties and give a
 procedure that selects the next child, with the appropriate probability
 distribution, using  $\polylog(n)$ random bits and  in $\polylog(n)$ time, and while
  increasing the space by  $\polylog(n)$. This procedure will be at the heart of our
  efficient implementation of $\nuc$.
 \ShortVersionEnd

\LongVersion
We first shortly discuss the challenges on the way to an efficient implementation of
\nuc. Consider the simple special case where the only two queries issued are,
 for some $j$, a
 single $\up(j)$ query followed by a single $\nuc(j)$ query (to simplify this discussion we assume that the the value of $k$ is globally known).
 Consider the situation after
   the query $\up(j)$.
   At this point, every node $x\in [j+1,n]$ may be a $u$-child of $j$ (namely, $u(x)$ may be set to $j$). 
In particular, since $\front(i)$ is initially $\nil$ for every $i$, it holds that $\varphi(x)=x-1$ and
$\pr[u(x)=j] = 1/(x-1)$.  Let $P_x$ denote the probability that node
$x$ is the first child of $j$. Then $P_x=\frac{1}{x-1}\cdot
\prod_{\ell=j+1}^{x-1} (1-\frac{1}{\ell-1}) = \frac{j-1}{(x-1)(x-2)}$ and for $P_{n+1}$ (i.e., $j$ has no child)
$P_{n+1}=\frac{j-1}{n-1}$.
{\color{black} As explained in the sequel,} each of the probabilities $P'_{k}=\sum_{x=j+1}^{k}P_x$ can be calculated in $O(1)$ time, therefore,
this random choice can be done in $O(\log n)$ time by choosing uniformly at random a number
 in $[0,1]$ and performing a binary search on $[j+1,n+1]$ to find which index it represents (see a more
  detailed and accurate statement of this procedure below).
However, in general, at the time of a certain \nuc\ query,
limitations may exist, due to previous queries, on the possible consistent values of  certain pointers $u(x)$.
There are two types of limitations:
\begin{inparaenum}[(i)]
\item $u(x)$ might have been already determined, or
\item $u(x)$ is still $\nil$ but  the option of $u(x)=i$ has been excluded since  $\front(i)>x$.
\end{inparaenum}
These limitations change the probabilities $P_x$ and $P'_x$, rendering them
more complicated and time-consuming to compute, thus rendering the above-defined  process
not efficient (i.e., not doable in $O(\log n)$ time). 
In the rest of this section we define and analyze a
 modified procedure that uses $\polylog(n)$ random bits, takes  $\polylog(n)$ time, and
  increases the space {\color{black} (that is used to store the state of the generator)} by  $\polylog(n)$. This procedure will
 be at the heart of the efficient implementation of \nuc.
\LongVersionEnd

\subsection{The invariants regarding the state of the generator}

\ShortVersion

In the implementation we maintain the following two invariant.
\ShortVersionEnd
\LongVersion

In the implementation of the on-the-fly generator of the pointers tree we will maintain two invariants that
are described below.
We will later discuss the cost  (in running time and space) of  maintaining  these invariants.
\LongVersionEnd

{\color{black} The purpose of these invariants is to obtain an efficient computation of the probabilities $P_x$ and $P'_x$ discussed in Subsection~\ref{ef.sec}.
Roughly speaking, if for every node $x \in \{1, \ldots , n-1\}$ it was the case that $\varphi(x+1) - \varphi(x) = 1$, then these probabilities were easy to calculate.
However, since the commitments of the generator impose changes on $\varphi$, this property can not hold for all the nodes.
The invariants ensure that the set of nodes for which this property does not hold are easy to identify and for which the $u$-parent is already determined.
}

\begin{invariant}\label{assume:father first}
  For every node $j$, the first $\nuct(j,\cdot,\cdot)$ query is always
  preceded by a $\up(j)$ query.
  \end{invariant}

We  will use this invariant
to infer that $\front(j)\neq \nil$ implies that  $u(j)\neq \nil$.
{\color{black} (recall that if a $\nuct(j,\cdot,\cdot)$ query was not issued thus far then $\front(j) =  \nil$).}
One can easily maintain this invariant  by introducing a $\up(j)$ query as the first step of
the implementation of the $\nuct(j,\cdot,\cdot)$ query (for technical reasons we do that in a lower-level
procedure \nuc.)

\begin{invariant}\label{inv:inv1}
  For every vertex $j$, $\front(j) \neq \mathsf{nil}$ implies that $\front( \front (j))
  \neq \mathsf{nil}$.
\end{invariant}

The second invariant is maintained by issuing an  ``internal''
$\nuc(\front(j),\front(j))$ query whenever $\front(j)$ is updated.
This is done recursively, the base of the recursion being node $n+1$.
\LongVersion
When analyzing the complexities of our algorithm we will take into account these recursive calls.
\LongVersionEnd
Let $\front^{-1}(j)$ denote the vertex $i$ such that
$\front(i)=j$, if such a vertex $i$ exists;
\LongVersion
(note that there can be at most one such node $i$, except for the case of $j=n+1$);
\LongVersionEnd
otherwise $\front^{-1}(j)=\nil$.
We get that  if $\front^{-1}(j)\neq \nil$, then $u(j)\neq\nil$.
See Figure~\ref{fig:rec_tree} for an illustrative example.

We note that if at a given time we consider a node $j$ such that $u(j)=\nil$ (i.e., its parent in the pointers tree is not yet
determined), then the set $\Phi(j)$ is the set of  all the nodes that can still be the parent of
node $j$ in the pointers tree. {\color{black} As mentioned above,} the set $\Phi$ is however defined also for nodes for which their parent is already
determined.

{\color{black} We are now ready to define the set $K$ which is, as we prove in the sequel, the set of nodes, $i$, for which $\varphi(i+1)=\varphi(i)$.
Moreover, we prove that if $x\notin K$ then $\varphi(i+1)-\varphi(i) = 1$. It is also the case that for $x\in K$ it holds that $u(x) \neq \nil$, as desired}
\begin{definition}
  Let $K$ denote
  the following set:
\ShortVersion
   $ K\triangleq\{i : \front(i) \neq \mathsf{nil} \text{ and }\front^{-1}(i)
    = \mathsf{nil}\}$.
\ShortVersionEnd
\LongVersion
  \begin{align*}
    K&\triangleq\{i : \front(i) \neq \mathsf{nil} \text{ and }\front^{-1}(i)
    = \mathsf{nil}\}~.
  \end{align*}
\LongVersionEnd
\end{definition}

\LongVersion
We shall prove the following lemma.

\begin{lemma}\label{lem:pf}
For any $x \in \{1, \ldots , n-1\}$:
\begin{equation}
\label{eq:aj2}
  \varphi(x+1) - \varphi(x) =
  \begin{cases}
    0& \text{if $x\in K$,}\\
    1& \text{if $x\notin K$}
  \end{cases}
  ~.
\end{equation}
\end{lemma}
\LongVersionEnd

\ShortVersion
The following lemma gives  properties of the series $\{\potenf(x)\}_x$. Proof omitted.
\ShortVersionEnd

Lemma~\ref{lem:pf} follows directly from the following more general claim.
\begin{claim}
  For every $x\in \{1, \ldots , n-1\}$:
\ShortVersion
\begin{inparaenum}[(1)]
\ShortVersionEnd
\LongVersion
\begin{enumerate}
\LongVersionEnd
\item $\potenf(x) \subseteq \potenf(x+1)\subseteq \potenf(x)\cup
  \{x,\front^{-1}(x)\}$.\label{item:pf1}
\item $\potenf(x+1) = \potenf(x)$ iff $x\in K$. \label{item:pf3}
\item  $\varphi(x+1) - \varphi(x) \leq 1$.\label{item:pf4}
\LongVersion
\end{enumerate}
\LongVersionEnd
\ShortVersion
\end{inparaenum}
\ShortVersionEnd
\end{claim}
\LongVersion
\begin{proof}

{\color{black}
We first observe that Item~\ref{item:pf1} follows directly from the definition of $\potenf$ and the properties of $\front$.
The see this, we first note that the fact that $\potenf(x) \subseteq \potenf(x+1)$ follows from that fact that for every $i$ such that $\front(i) < x$ it clearly holds that $\front(i) < x+1$.
On the other hand, there might be only a single node, $i$, such that $\front(i) < x+1$ but $\front(i) \geq x$. This is possible only when $\front(i) = x$, which might be the case only for a single node, which is the $u$-parent of $x$.
Finally, $x\in \potenf(x+1)$ if and only if $\front(x) = \nil$.

  To prove Item~\ref{item:pf3}, observe that by Item~\ref{item:pf1}, $\potenf(x)=\potenf(x+1)$,
  iff $x\notin\potenf(x+1)$ and
  $\front^{-1}(x)\notin\potenf(x+1)$.
  As stated above, $x\in\potenf(x+1)$ iff $\front(x) = \nil$.
  Similarly, $\front^{-1}(x)\notin\potenf(x+1)$ iff $\front^{-1}(x) = \nil$.

}
Finally, to prove Item~\ref{item:pf4} we need to show that it is not possible for
both $x$ and $\front^{-1}(x)$ to belong to $\potenf(x+1)$.  Indeed, if
$\front^{-1}(x)\in \potenf(x+1)$, then there exists a vertex $i$ such that
$\front(i)=x$. Invariant~\ref{inv:inv1} implies that $\front(x)=\front(\front(i))\neq \nil$.
However, $x\in\potenf(x+1)$ implies $\front(x)=\nil$, a contradiction.
\end{proof}
\LongVersionEnd

\ShortVersion
We now describe the implementation of  $\nuct(j,k,flag)$ and $ \nuc(j)$. $\nuct(j,k,flag)$ is
a loop of  $\nuca(j,k)$ until the right flag is found, and $\nuca(j,k)$ is essentially   a call to $\nuc(j)$ (see
Figure~\ref{alg:nuc}).
Note  that if $j$ does not have children larger than $k$, then $\nuca(j,k)$ returns $n+1$.
\ShortVersionEnd

\subsection{Description of the efficient implementation}
\LongVersion
We are now ready to describe the implementation of  $\nuct(j,k,flag)$ and $ \nuc(j)$.
The state of the generator is stored using the following data structures, of which the implementation of \nuc\  (and of \up) makes use of.
\ShortVersion
\begin{inparaenum}[(1)]
\ShortVersionEnd
\LongVersion
\begin{itemize}
\LongVersionEnd
\item An array of length $n$, $u(j)$.
\item An array of length $n$, $flag(j)$.
\item An array of length $n$, $\front(j)$ (We also maintain an array $\front^{-1}(i)$ with the natural definition).
 \item An array  of $n$ balanced search trees, called $\children(j)$, each  holding the set of nodes $i >j$ such that $u(i)=j$.
 {\color{black}The operations we use on the search trees are insert and successor, where $\childinsert(T, j)$ inserts the element $j$ to the tree $T$ and $\childsccsr(T, j)$ returns the smallest elements in $T$ which is larger than $j$.}
 For technical reasons all trees $\children(j)$ are initiated with  $n+1 \in \children(j)$.
\item A number of additional data structures that are implicit in the listing,
described and analyzed in the sequel.
\LongVersion
\end{itemize}
\LongVersionEnd
\ShortVersion
\end{inparaenum}
\ShortVersionEnd

As seen in  Figure~\ref{alg:nuct}, $\nuct(j,k,flag)$ is merely a loop of  $\nuca(j,k)$, and $\nuca(j,k)$ is essentially
 a call to $\nuc(j)$.
The ``real work'' is done in the implementation of $\nuc(j)$ that we describe now.
Note  that if $j$ does not have
 children larger than $k$, then $\nuca(j,k)$ returns $n+1$. 
\LongVersionEnd

If $\front (j) > k$  when $\nuca(j,k)$ is called, then the next child is already fixed and it is
just extracted from the data structures.

Otherwise,
 an interval $I=[a,b]$ is defined, and it will contain the answer of $\nuc(j)$.
Let $a = \front(j)+1$ if $\front(j)\neq \nil$; and
$a=j+1$, if $\front(j)= \nil$.
Let $b= \min\{ \{\ell > \front(j): u(\ell)=j\}\cup \{n+1\}\}$ if $\front(j)\neq \nil$;
and $b= \min\{ \{\ell > j:  u(\ell)=j\}\cup \{n+1\}\}$,
if $\front(j) = \nil$
Observe that
no vertex $x\in F\cap [a,b)$ can satisfy $u(x)=j$. Hence, the answer is in
$I\setminus (F\setminus \{b\})$.

The next child can be sampled according to the desired distribution in
 a straightforward way by going sequentially over the vertices in $I\setminus (F\setminus \{b\})$,
  and tossing for each vertex $x$ a coin that has  probability $1/\varphi(x)$ to be $1$, until indeed one of those coins
  comes out $1$,  or all vertices are exhausted (in which case node $b$ is taken as
   the next child).
   We denote by $D(x)$, $x \in I\setminus F$, the probability that $x$ is chosen when the
   the above procedure is applied.
   This  procedure, however, takes linear time.

  \ShortVersion
In order to start building our efficient implementation for \nuc\ we consider  the same
process, with the same probabilities $1/\varphi(x)$, but this time for  $[a,b)\setminus K$,
rather than $[a,b)\setminus F$.
\ShortVersionEnd
 \LongVersion
  In order to start building our efficient implementation for \nuc\
  we note that by the definition of $K$, $K \subseteq F$,
and we consider a process where
we toss  $|[a,b)\setminus K|$ coins
 sequentially for the vertices in $[a,b)\setminus K$.  The probability that the coin for $x\in [a,b)\setminus K$
is $1$ is still $1/\varphi(x)$. We stop as soon as $1$ is encountered or on $b$ if all
coins are $0$.
\LongVersionEnd
The vertex on which we stop, denote it $x$, is a {\em candidate next $u$-child}. If $x\in F\setminus K \setminus \{b \}$,
then $x$ cannot be a child of $j$ {\color{black} (because it already has a $u$-parent)} so we proceed by repeating the same process {\color{black} recursively}, but with the interval $[x+1,b]$
instead of the interval $[a,b]$.
\LongVersion
We denote by $D'(x)$, $x \in I\setminus F$, the probability that $x$ is chosen when this procedure is applied.
\LongVersionEnd

\subsubsection{Efficiently selecting a $u$-child candidate}
We  now build our efficient procedure that selects the candidate, without sequentially going over the nodes.
To this end, observe that the sequence of probabilities of the
coins tossed in the last-described process behaves ``nicely''.  Namely, the probabilities
$1/\varphi(x)$, for $x \in [a,b) \setminus K$,
form the harmonic sequence  starting from $1/\varphi(a)$
and ending in $1/(\varphi(a)+|[a,b)\setminus K|-1)$.
\ShortVersion
Indeed, Lemma~\ref{lem:pf}
\ShortVersionEnd
\LongVersion
Indeed, Eq.~\eqref{eq:aj2}
\LongVersionEnd
implies that if vertex $i$ is the smallest vertex in $I\setminus K$, then
$\varphi(i)=\varphi(a)$ and an increment between $\varphi(x)$ and
$\varphi(x+1)$ occurs if and only if $x\notin K$.
Let $s= | [a,b) \setminus K|$ and let $P_h$,
$ 0 \leq h \leq s$ be the probability that the node
of rank $h$ in $([a,b) \setminus K) \cup \{b \}$ is chosen as candidate in the sequential procedure defined above.
Since $1/\varphi(x)$ forms the  harmonic sequence for $x\in [a,b)\setminus K$,
we can, given $\varphi(a)$,
 calculate in $O(1)$ time,  for any $0 \leq i \leq s+1$, the probability $P'_i= \sum_{q < i} P_q$ (i.e., the probability that a node of some rank
$q$, $q< i$, is chosen).   Indeed, for $i=0$, $P_i=\frac{1}{\varphi(a)}$;
 for  $0 < i < s$,  $P_i=\frac{1}{\varphi(a)+i}\cdot \prod_{\ell=0}^{i-1}
\left(1-\frac{1}{\varphi(a)+\ell}\right) = \frac{\varphi(a)-1}{(\varphi(a)+i-1)(\varphi(a)+i)}$;
  and for $i=s$, $P_{s}=\prod_{\ell=0}^{s-1} \left(1-\frac{1}{\varphi(a)+\ell}\right) =
\frac{\varphi(a)-1}{\varphi(a)+s-1}$.
Hence,  for $0 \leq  i \leq  s$,
$P'_i = 1-\frac{\varphi(a)-1}{\varphi(a)+(i-1)}$,
and for $i=s+1$, $P'_{s+1}=1$.
This allows us  to simulate   one iteration  (i.e., choosing the next  {\em candidate next $u$-child}) by choosing
uniformly  at random a single number in $[0,1]$, and then performing a binary search over $0$ to $s$   to
decide what rank $h$  this number  ``represents''.
After the  rank $h \in [0,s]$ is selected, $h$ is then mapped to the vertex of rank $h$ in  $([a,b) \setminus K) \cup \{b \}$,
denote it $x$, and this is the {\em candidate next $u$-child}. As before,
if $x\in F\setminus K \setminus \{b \}$, then $x$ cannot be a child of $j$ so we ignore it and proceed in the same way,
this time with the interval $[x+1,b]$.
We denote by $\hat{D}(x)$, $x \in I\setminus F$ the probability that $x$ is chosen  when
   this  third procedure is applied.
   See~Figure~\ref{alg:toss} for a formal definition of this procedure and that of \nuc.

 Observe that this
procedure
takes $O(\log s)$ time (see  Section~\ref{sec:complexities} for a formal statement of the time
and randomness complexities).
We note that  we cannot perform this selection procedure in the same time complexity for
the set $[a,b)\setminus F$, because
we do not have a way to calculate each and every probability $P'_i$, $i \in [a,b)\setminus F$, in $O(1)$ time,
even if $\varphi(a)$ is given.

To conclude the description of the implementation of $\nuc$, we give the following lemma which states that
 the probability distribution on the next child is the
  same for all three processes described above.
    \ShortVersion
  The (technical) proof is omitted.
  \ShortVersionEnd
 \begin{lemma}
 \label{le:equal_prob}
 For all $x \in I\setminus F$,  $\hat{D}(x)={D}(x)$.
 \end{lemma}
 \LongVersion
 \begin{proof}
 To prove the claim we prove that $\hat{D}(x)={D}'(x)$ and that ${D}'(x)={D}(x)$.

 To prove the latter, denote by $x_1<x_2<\ldots<x_k$ the nodes in the set $I\setminus F$, where $k=|I\setminus F|$,
 and let $p(x_j)=\frac{1}{\varphi(x_j)}$.
 For any $1 \leq j \leq k-1$ $D(x_j)=p(x_j)\cdot \Pi_{i=1}^{j-1}(1-p(x_i))$, and for $x_k$ (which is
 the node denoted $b$ in the discussion above), $D(x_k)=1- \Pi_{i=1}^{k-1}(1-p(x_i))$.

 When we consider the sequential process where one tosses a coin sequentially for all
  nodes in $I \setminus K$ (and not only for the nodes in $I \setminus F$) we extend the
   definition of $D'(\cdot)$ to be defined also for nodes in $ I \setminus  K$. For a node
   $z \in (I \setminus  K) \cap F$, $D'(z)$ is the probability that $x$ is chosen as a
   {\em candidate next $u$-child}. Thus, if we denote by $y_1<y_2<\ldots<y_{\ell}$,
   $\ell=|I\setminus K|$, the nodes in $ I \setminus  K$ we have that
   $D'(y_j)=p(y_j)\cdot \Pi_{1 \leq i < j; y_j \in I\setminus F}(1-p(y_i))$, and for $y_{\ell}$ (which is
 the node denoted $b$ in the discussion above), $D(y_{\ell})=1- \Pi_{1 \leq i < \ell; y_j \in I\setminus F}(1-p(y_i))$.
   Thus, for any $x \in  I \setminus F$, $D(x)=D'(x)$.

We now extend $\hat{D}(\cdot)$ to be defined for all nodes in
    $I \setminus K$.  The assertion $\hat{D}(x)={D}'(x)$,  for any $x \in I \setminus K$, follows from the fact a
    number  $ M \in [0,1]$ is  selected uniformly
     at random and then the interval in which it lies is found.
     That is, $i$ is selected if and only if $ P'_i \leq M < P'_{i+1} $ which, by the definitions of $P_i$ and $P'_i$, occurs
     with probability $P_i=D'(x_i)$.
 \end{proof}
 \LongVersionEnd

\begin{figure}
\begin{minipage}{0.44\textwidth}
\begin{tcolorbox}[title ={\footnotesize \nuct \\ Returns the least $i >k$, $i$  is a  u-child of $j$, $i$ has flag ``flag''. \\ Assumes that $k \leq \front(j)$}]
 \begin{algorithmic}[1]
{\footnotesize
 \Procedure{$\nuct$}{$j,k,flag$}
 	\State $x \gets k$
	 \Repeat
	      \State $ x \gets \nuca(j,x) $
	    \Until{$flag(x) = flag$ \mbox{or} $x=n+1$}
  	   \State \textbf{return} $x$
 \EndProcedure
}
\end{algorithmic}
\end{tcolorbox}
\begin{tcolorbox}[title ={\footnotesize \nuca \\ Returns the least $i \geq k $, $i$  is a  u-child of $j$.  \\ Assumes that $k \leq \front(j)$.}]
\begin{algorithmic}[1]
{\footnotesize
 \Procedure{$\nuca$}{$j,k$}
        \State If ($k \geq n$) \textbf{return}  $(n+1)$
	\State $q \leftarrow \childsccsr(\children(j),k)$
	\If {$q \leq \front(j)$}
	   \State \textbf{return}  $q$
	\Else
	      \State \textbf{return}  $\nuc(j)$
         \EndIf
\EndProcedure
}
\end{algorithmic}
\end{tcolorbox}
\begin{tcolorbox}[title ={\footnotesize \up \\ Returns the $u$-parent of  $j$.}]
 \begin{algorithmic}[1]
{\footnotesize
 \Procedure {$\up$}{$j$}
 \If {$u(j) = \nil$}
 	\State $u(j) \gets_{R} [1,j-1]$
         \State $flag(j) \leftarrow_{R}\{\direct,\recursive\}$
         \State $\childinsert(\children(u(j)),j)$
           \EndIf
 \State  \textbf{return} $(u(j),flag(j))$
\EndProcedure
}
\end{algorithmic}
\end{tcolorbox}
\end{minipage}
\caption{Pseudo code  of the pointers tree generator (part 1)
\label{alg:nuct}
\label{alg:u-parent}
}
\end{figure}

\begin{figure}
\begin{minipage}{0.59\textwidth}
\begin{tcolorbox}[title ={\footnotesize \nuc  \\ Returns  the least $i  > \front (j) $  which is a  u-child of $j$.}]
 \begin{algorithmic}[1]
{\footnotesize
 \Procedure{$\nuc$}{$j$}
 		 \State $(p,t) \gets \up(j)$
                   \State If ($\front(j) \geq n$) \textbf{return}  $(n+1)$
 		  \State $a \gets
				\begin{cases}
				  \front(j)+1\ &\text{if $ \front(j)\neq\nil$}\\
				  j+1& \text{if $\front(j)=\nil$}
				\end{cases}
			$
		\State 	$b \gets
		                 \begin{cases}
				 \childsccsr(\children(j),\front(j))  &\text{if $\front(j)\neq\nil$}\\
				 \childsccsr(\children(j),j)& \text{if $ \front(j)=\nil$}
				\end{cases}	
				$	
		  \Repeat
			\State $s\gets |[a,b) \setminus K|$
			\State $h \gets \toss (\varphi(a),s+1)$
			 \If{$h=s$}
			      \State  \textbf{return} $b$
			 \Else	
			        \State $x\gets $ the vertex of rank $h$ in $[a,b)\setminus K$
			        \If{$u(x)=\nil$}
			           \State $u(x)=j$
			           \State $flag(x) \leftarrow_{R}\{\direct,\recursive\}$
			           \State $\childinsert(\children(j),x)$
			           \State $\front(j)\gets x$
			           \State  $\front^{-1}(x) \gets j$
          			   \State if ($\front(x)=\nil$) $\nuc(x)$
			           \State \textbf{return} $(x)$
			       \Else $~~$     /* i.e., if $u(x)\neq \nil$ */
			           \State $a\gets x+1$
                                \EndIf  	
                          \EndIf
                 \Until{forever}
\EndProcedure
}
\end{algorithmic}
\end{tcolorbox}
\end{minipage}
\begin{tcolorbox}[title = {\footnotesize \toss  \\ Returns a random rank $0 \leq y \leq t-1$.}]
 \begin{algorithmic}[1]
{\footnotesize
  \Procedure {$\toss$}{$\xi,t$}
	 \State  $\alpha \gets n^c$ {\tt (for some constant $c>1$)}.	
	 \State Choose uniformly at random  an integer  $M \in [0,\alpha]$
	 \State $H \gets M\cdot \frac{1}{\alpha}$
 	 \State Using binary search on  $[0,t-1]$ find $0 \leq y \leq t-1$  such that
 			 $P'_y \leq H  < P'_{y+1}$  \\
			     \hspace{3cm} {\tt (where, for $0 \leq y \leq t-1$, $P'_y=1 -\frac{\xi-1}{\xi+(y-1)}$, and $P'_t=1$) }
	 \If{$(H+1)\frac{1}{\alpha} \leq P'_{y+1}$\label{if.line}}	
	    \State \textbf{return}  $y$
	\Else
	    \State  $\alpha \gets \alpha \cdot\Pi_{y=0}^{t-1}(P'_{y+1}-P'_{y})$\label{line.10}
	    \State Choose uniformly at random  an integer  $M \in [0,\alpha]$
	    \State $H \gets M\cdot \frac{1}{\alpha}$
	    \State Using binary search on  $[0,t-1]$ find $0 \leq y \leq t-1$  such that
 			 $P'_y \leq H < P'_{y+1}$ \\
			     \hspace{3cm} {\tt (where, for $0 \leq y \leq t-1$, $P'_y=1 -\frac{\xi-1}{\xi+(y-1)}$, and $P'_t=1$) }
	    \State \textbf{return}  $y$
	\EndIf
 \EndProcedure
}
\end{algorithmic}
\end{tcolorbox}
\caption{Pseudo code  of the pointers tree generator  (part 2)
\label{alg:nuc}
\label{alg:toss}
}
\end{figure}

\LongVersion
\subsection{Implementation of \up}
\LongVersionEnd

The implementation of \up\ is  straightforward
 (see~Figure~\ref{alg:u-parent}).
However,  note that updating the
 various data structures,  while  implicit in the listing, is  accounted for in the time
 analysis.

\LongVersion
\subsection{Analysis of the pointer tree generator}
\LongVersionEnd
\ShortVersion
\subsection{Analysis of the pointer tree generator}
\ShortVersionEnd
\label{sec:complexities}

We first give the following claim that we later use a number of times.
\begin{lemma}
\label{le:recursive_calls}
With high probability,  for each and every 
call to $\nuc$, the  size of the recursion tree of  that call,  for calls to \nuc, is $O(\log n)$.
\end{lemma}
\begin{proof}
Consider the recursive invocation tree that results from a call to \nuc. Observe that (1) by the code of \nuc\
  this tree is in fact a path; and (2) this path corresponds to a path in the pointers tree, where each
  edge of this tree-path is ``discovered'' by the corresponding  call to \nuc. That is, the maximum size of a recursion tree
of a
call of \nuc\ is bounded from above by the height of the pointers tree. By Claim~\ref{cl:low_degree_and_height},
with high probability, this is $O(\log n)$.
\end{proof}

{\color{black}

\subsection{Efficient rolling of dice} \label{toss.sec}
In Algorithm~\ref{alg:toss} we implement a rolling of a dice whose time complexity is, with high probability, $O(\log n)$, as explained next.
The commulative probabilities $P'_y$ of each side of
the dice are a function of $\xi$ and $y$ and are computable in
$O(1)$ time. To determine the outcome of a roll of the dice, we
pick $r = c\log n$ random unbiased bits, which are interpreted as a
binary representation of a subinterval $[\ell_1,\ell_2]$ of $[0,1]$ of
length $2^{-r}$.  We say that the subinterval is good if it is
contained in an interval, the endpoints of which are consecutive
cumulative probabilities. Namely, $[\ell_1,\ell_2]$ is good if there
exists an $i$ such that $[\ell_1,\ell_2)\subseteq [P'_i,
P'_{i+1}]$. Note that in this case, we choose $i$ as the
outcome of the roll of the dice. The number of bad subintervals is
bounded by $t$, which is at most $n$. Hence, the probability that the
subinterval is bad is at most $n\cdot 2^{-r}$. Since $r=c\log n$,
the probability of determining the side of the dice after
$r$ random bits is at least $1-n^{-c+1}$.
Hence, the time complexity in this case is dominated by the time complexity of the random search which is $O(\log n)$.
On the other hand, if the interval we picked is bad (which happens with negligible probability), then we refine the size of the intervals so that all intervals are good (Line~\ref{line.10}).
In this case the time complexity is $O(\max\{t, \log n\})$.

We note that by a standard
technique we could alternatively implement a Las-Vegas algorithm that rolls the dice.
In this case the algorithm continues refining the intervals (by
using additional random bits, one at a time) until the subinterval is good.
Namely, the Las-Vegas algorithm, after each random bit, checks in time $O(\log n)$ if the
subinterval is good by performing a binary search over the commulative
probabilities and adds an additional bit only if the selected interval is bad.

}

\subsubsection{Data structures and space complexity}
\label{sec:space_complexity}
The efficient implementation of \nuc\ makes use of the following data structures.
\begin{itemize}
\item A  number  of arrays of length $n$, $u(j)$ and $flag(j)$, $\front(j)$ and $\front^{-1}(j)$, used to
 store various values  for nodes $j$.
Since we implement arrays by means of search trees, the space complexity of each array is $O(m)$, where $m$ is
 the maximum number of distinct keys stored with a non-null value in that array, at any given time.
 The time
complexity for each
operation on this arrays is $O(\log m)=O(\log n)$ (since they are implemented as balanced binary search trees).

\item For each node $j$, a balanced binary search tree called $\children(j)$,  where $\children(j)$ includes all
 nodes $i$ such that $u(i)=j$ (for technical reasons we define $\children(j)$ to always include node $n+1$.)
\LongVersion
\footnote{So that we maintain low space complexity, for a given $(j)$,
 $\children(j)$ is initialized only at the first use of $\children(j)$, at which time node $n+1$ is inserted .}
\LongVersionEnd
 Observe that for  each  child $i$ stored in one of these trees, $u(i)$  is already determined.
 Thus, the increase, during a given period, in the space used by  the  $\children$ trees is bounded
  from above by the  the number of nodes $i$ for which $u(i)$ got determined during that period.
  For the time complexity of the operations on these trees
  we use a  coarse standard upper bound of $O(\log n)$ on each tree operation.\footnote{In fact we can use the fact that with high
  probability $C_j=O(\log n)$ and
  get, with high probability, an upper bound of   $O(\log C_j)=O(\log \log n)$ on the  time complexity.}

 We store the roots of all non-empty trees $\children(j)$ in an ``array''. Thus, using our implementation of
 arrays as balanced search trees, the space used by this ``array'' is $O(m)$ and the time
  to  access the root of a certain $\children(j)$
is $O(\log m)=O(\log n)$, where $m$ is the number of non-empty trees $\children(j)$ at a given time.

\end{itemize}

The listings of the implementations of the various procedures leave {\em implicit}  the maintenance of  two data structures,   related
to the set $K$ and to the computation of $\varphi(\cdot)$:
\begin{itemize}
\item A data structure that allows one to retrieve the value of $\varphi(a)$ for a given node $a$.  This data structure is implemented
by retrieving the cardinality of {\color{black} the set of nodes that are not potential parent of $a$, i.e., $(a-1)-\varphi(a)$,} for a given node $a$.
The latter is equivalent to counting how many nodes $i <a$ have
$\front(i)\neq \nil$ and $\front(i) \geq a$.
We use two balanced binary search trees (or order statistics trees) in a specific way and have that
by standard implementations of balanced search trees
 the space complexity is $O(k)$ (and all operations are done in time $O(\log k)=O(\log n)$). Here $k$ denotes the number of nodes $i$ such that  $\front(i)\neq \nil$.
\ShortVersion
The details of the implementation are omitted from this extended abstract.
\ShortVersionEnd
\LongVersion
 More details of the implementation of this data structure appear in the appendix (See Section~\ref{se:stabbing_data_structure}).
\LongVersionEnd

\item A data structure that  allows one to find the vertex of rank $h$ in the ordered set
$[a,n+1]\setminus K$. This data structure is implemented  by a balanced binary search tree storing the nodes
 in $K$,
augmented with the queries $\trank_K(i)$ (as in an order-statistics tree{\color{black}~\cite[Sec.~14]{cormen2009introduction}~\footnote{{\color{black}An order-statistic tree is a data structure, that supports two operations beyond the classical balanced binary search tree operations (i.e., insertion, lookup, and deletion), as follows. Informally, given an index $i$ the $\tselect(i)$ operation returns the $i$-th element in the sorted list of the elements that are in the tree. On the other hand, the rank of an element $x$ in the tree is its index in that sorted list.}}}) as well as $\trank_{\bar{K}}(i)$ and $\tselect_{\bar{K}}(s)$, i.e., finding
the element of rank $s$  in the complement of $K$.  To find the vertex of rank $h$ in $[a,n+1]\setminus K$
we use the query $\tselect_{\bar{K}}(\trank_{\bar{K}}(a)+h)$.
 The space complexity of this data structure is $O(k)$,
 and all operations are done in time $O(\log k)=O(\log n)$ or $O(\log^2 k)=O(\log^2n)$ (for the $\tselect_{\bar{K}}(i)$ query).
  Here $k$ denotes the number of nodes in $K$, which is upper bounded by the number of nodes $i$ such that  $\front(i)\neq \nil$.
\ShortVersion
The details of the implementation are omitted from this extended abstract.
\ShortVersionEnd
\LongVersion
More details of the implementation of this data structure appear in the appendix (See Section~\ref{se:rank_data_structure}).
\LongVersionEnd

\end{itemize}

\subsubsection{Time complexity}

\ShortVersion
\noindent {\bf Time complexity of $\toss(\varphi,s$).}
\ShortVersionEnd
\LongVersion
\paragraph{Time complexity of $\toss(\varphi,s$).}
\LongVersionEnd
{\color{black} The time complexity of
this procedure
 is with high probability $O(\log n)$ (see Section~\ref{toss.sec})}

\ShortVersion
\noindent {\bf Time complexity of ``$x \gets$ the vertex of rank $h$ in $[a,n+1]\setminus K$''.}
\ShortVersionEnd
\LongVersion
\paragraph{Time complexity of ``$x \gets$ the vertex of rank $h$ in $[a,n+1]\setminus K$''.}
\LongVersionEnd
This operation is implemented using the data structure defined above,
and takes  $O(\log^2n)$ time.

\ShortVersion
\noindent {\bf Time complexity of $\up(j)$.}
\ShortVersionEnd
\LongVersion
\paragraph{Time complexity of $\up(j)$.}
\LongVersionEnd
{\color{black} As stated in Lemma~\ref{le:complexity_pointers_tree}, the time complexity of \parent\ is $O(\log n)$.}

\ShortVersion
\noindent {\bf Time complexity of  \nuc.}
\ShortVersionEnd
\LongVersion
\paragraph{Time complexity of  \nuc.}
\LongVersionEnd
First consider the time complexity  consumed by  a single invocation of \nuc\ (i.e., without taking into account the time
 consumed by recursive calls of \nuc):\footnote{We talk about an ``invocation'', rather than a ``call'', when we want to emphasize that
  we consider only the resources consumed by a single level of the recursion tree.}
The call to \up\ takes $O(\log n)$ time.
Therefore, until the start of the repeat loop, the  time is $O(\log n)$ (the time complexity of  \childsccsr\  is  $O(\log n)$).
Now, the time complexity of a single iteration of the loop (without taking into account  recursive calls to \nuc) is
 $(O \log^2 n)$ because:
\begin{itemize}
\item Each access to an ``array'' takes $O(\log n)$ time.
\item Calculating $\varphi(a)$ takes $O(\log n)$ time.
\item The call to \toss\ takes  $O(\log n)$ time.
\item Finding the vertex of rank $h$ in $[a,n+1]\setminus K$ takes $O(\log^2 n)$ time.
\item Each of the $O(1)$ updates of $\front(\cdot)$ or $\front^{-1}(\cdot)$ may change the set $K$, and therefore may
 take $O(\log n)$ time to update the data structure involving $K$.
 \item An update of  any given $\children(\cdot)$ binary search tree takes $O(\log n)$ time.
\end{itemize}

We now examine the number of iterations of the loop.
\begin{claim}
\label{cl:number_of_iterations}
With high probability, the number of iterations of the loop in a single invocation of \nuc\ is $O(\log n)$.
\end{claim}
\begin{proof}
We consider a process where the iterations continue
until the selected node is node $b$. A random variable, $R$,  depicting this number  dominates a random
 variable that depicts the actual number of iterations.
For each iteration, an additional node is selected by \toss. By Lemma~\ref{le:equal_prob} the probability that a node
$j<b$ is selected by \toss\ is  $1/\varphi(j)$, and we have that $1/\varphi(j) \leq \frac{1}{j-1}$.
Thus, $R=1+\sum_{j=a}^{b-1}X_j$, where $X_j$ is $1$ iff  node $j$ was selected, $0$ otherwise.
Since $\mu=\sum_{j=a}^{b-1}\frac{1}{\varphi(j)} \leq \log n$,
using Chernoff bound
\LongVersion
\footnote{cf.~\cite{guidedtour}, Inequality (8).}
\LongVersionEnd
we have, for any constant $c>6$,
$P[R > c \cdot \log n] \leq 2^{-c\cdot \log n}=n^{-\Omega(1)}$.
\end{proof}

\noindent We thus have the following.
\begin{lemma}
\label{le:time_single_nuc}
For any given invocation of \nuc, with high probability, the time complexity is  $O(\log^3 n)$.
\end{lemma}

\subsubsection{Randomness complexity}

\LongVersion
Randomness is used in our generator to randomly select the parent of the nodes (in \up)
 and to randomly select a next child for a node (in \toss). We use the common convention that, for any given $m$, one can
 choose uniformly at random
an integer in  $[0,m-1]$ using $O(\log m)$ random bits and in  $O(1)$ computation time.
  We give our algorithms and analyses based on this building block.
\LongVersionEnd

In procedure \up\ we use $O(\log n)$ random bits whenever, for a given $j$,  this procedure is called with parameter $j$ for the first time.

In procedure \toss\ the {\tt if} condition holds with probability $1-1/n^{c-1}$ (where $c$ is the constant used in that procedure).
Therefore, given a call to \toss,
 with probability $1-1/n^{c-1}$ this procedure uses $O(\log n)$ bits.
By Claim~\ref{cl:number_of_iterations}, in each call to \nuc\ the number of times that  \toss\ is called is, \whp,
$O(\log n)$. We thus have the following.

\begin{lemma}
\label{le:single_nuc_rand}
During a given call to  \nuc, w.h.p., $O(\log ^2 n)$ random bits are used.
\end{lemma}

 The following lemma states the time, space, and randomness complexities of the queries.
 \begin{lemma}
\label{le:complexity_pointers_tree}
The complexities of \nuct\ and \up\ are as follows.
\begin{itemize}
\item
Given a call to  \up\ the following hold for this call:
\ShortVersion
\begin{inparaenum}[(1)]
\ShortVersionEnd
\LongVersion
\begin{enumerate}
\LongVersionEnd
\item The increase, during the call, of the space used by our algorithm is $O(1)$.
\item The number of  random bits used during that call is $O(\log n)$.
\item The time complexity of that call   is $O(\log n)$.
\LongVersion
\end{enumerate}
\LongVersionEnd
\ShortVersion
\end{inparaenum}
\ShortVersionEnd
\item
 Given an call to  \nuct, with high probability, all of the following hold for this call:
\ShortVersion
\begin{inparaenum}[(1)]
\ShortVersionEnd
\LongVersion
\begin{enumerate}
\LongVersionEnd
\item The increase, during that call, of the space used  by our algorithm is $O(\log^2 n)$.
\item The number of  random bits used during that call is $O(\log^4 n)$.
\item The time complexity of that call   is $O(\log^5 n)$.
\LongVersion
\end{enumerate}
\LongVersionEnd
\ShortVersion
\end{inparaenum}
\ShortVersionEnd
\end{itemize}
 \end{lemma}
 \begin{proof}

During a call to $\up(j)$ the size of the used space increases when a pointer $u(j)$  becomes non-null
 or when additional values are stored in $\children(u(j))$. To select $u(j)$, $O(\log n)$ random bits are used,
  and $O(\log n)$ time is used to insert $j$ in $\children(u(j))$ and to update the data
  structure for the set $K$  (this is implicit in the listing).

 For the analysis of \nuct, we first consider \nuc.
 Observe that by  Lemma~\ref{le:recursive_calls}, \whp, each and every root (non-recursive) call of \nuc\ has
a recursion tree of size $O(\log n)$.  
  In each invocation of \nuc, $O(1)$  variables $\front(j)$ and $u(j)$ may be updated. Therefore, \whp, for all root (non-recursive) calls to \nuc\ it holds that the
  increase in space during this call is $O(\log n)$ (see Section~\ref{sec:space_complexity}).  
  Using Lemmas~\ref{le:single_nuc_rand} and~\ref{le:recursive_calls}
  we have that, \whp, each root  call  of \nuc\   uses  $O(\log^3 n)$ random bits. 
Using Lemmas~\ref{le:time_single_nuc} and~\ref{le:recursive_calls},   we have that, \whp, the
   time complexity of each root call of \nuc\  is $O(\log^4 n)$.

  Because the flags of the pointers are uniformly distributed in $\{\direct,\recursive\}$,
  each call
  to \nuct\ results, \whp, in $O(\log n)$ calls to \nuc. The above complexities are thus multiplied
  by an $O(\log n)$  factor to get the
  (w.h.p.)
  complexities of \nuct.
 \end{proof}

\section{On-the-fly Generator for BA-Graphs}

\begin{figure}
\begin{minipage}{0.65\textwidth}
\begin{tcolorbox}[title ={\footnotesize \nn \\  Returns the next  neighbor of $j$ in the BA-graph.}]
 \begin{algorithmic}[1]
{\footnotesize
 \Procedure{$\nn$}{$j$}
 \If{$first\_query(j) = \consttrue $}    \\
 						   /* first query for $j$ */
  	\State $first\_query(j) \gets \constfalse$
	\State $\hinsert(heap_j,n+1)$
	\State  $\hinsert(heap_j, \nuct(j,j,\direct))$
  	\State \textbf{return} $\fparent(j)$
 \Else  \\
 		 /* all subsequent queries for $j$ */
 	 	\State $r \leftarrow \hextmin(heap_j)$
    		 \If {$r = n+1$}
	  	         \State $\hinsert(heap_j,n+1)$
    			\State \textbf{return} $n+1$
       		\Else
			   \If {$flag (r) = \direct $}
				   \State $   \hinsert(heap_j,\nuct(j,r,\direct))$
				    \State $   \hinsert(heap_j,\nuct(r,r,\recursive))$
			   \Else
			        \State $(q,flag) \leftarrow  \up(r)$
     			 	\State $ \hinsert(heap_j,\nuct(q,r,\recursive))$
				 \State $   \hinsert(heap_j,\nuct(r,r,\recursive))$
 		 	   \EndIf
 		            \State \textbf{return} $r$
        \EndIf
\EndIf
\EndProcedure
}
\end{algorithmic}
\end{tcolorbox}
\end{minipage}
\hfill
\begin{minipage}[t]{0.39\textwidth}
\begin{tcolorbox}[left skip=2pt, title = {\footnotesize \fparent  \\ Returns the parent of $j$ in the  \\ BA-graph.}]
 \begin{algorithmic}[1]
{\footnotesize
\Procedure{$\fparent$}{$j$}
\State $(i,flag) \leftarrow \up(j)$
\If {$flag=\mbox{\direct}$}
	\State \textbf{return} $i$
\Else
 	 \State \textbf{return} $\fparent(i)$
\EndIf
\EndProcedure
}
\end{algorithmic}
\end{tcolorbox}
\end{minipage}
\caption{Pseudo code of the on-the-fly BA generator
\label{alg:next_neighbor}
\label{alg:find-parent}
}
\end{figure}

 Our  on-the-fly generator for BA-graphs {\color{black} (see definition in Section~\ref{section.models})} is called  \flyBA,
 and simply
calls $\nn(v)$  for each query  on node $v$. 
We present an implementation for the \nn\  query, and
 prove its correctness,  as well  as analyze its time, space, and
 randomness complexities.
 The on-the-fly BA generator maintains
 $n$ standard {\tt heaps}, one for each node. The heaps store nodes,  where the order
is the  natural order
 of their serial numbers.
 \LongVersion
 \footnote {For simplicity of presentation we assume that the initialization of the heap occurs at the first insert, and make sure in our use of the heap that no extraction is performed before the first insert.}
 \LongVersionEnd
The  heap of  node $j$  stores some of the nodes already known
 to be neighbors of $j$. 
\LongVersion
 In addition, the generator maintains for  purely technical reasons an array of size $n$,  $first\_query$, indicating
 if a \nn\ query has  been issued for a given node. 
   The implementation of  the \nn\  query works as follows  (see~Figure~\ref{alg:next_neighbor}).
 \LongVersionEnd
  \begin{itemize}
\item  {\em For the  first} \nn($j$) {\em query}, for a given $j$, we proceed as follows.
    We find the parent of $j$
   in the {\color{black} constructed} BA-graph, which is done by
   following, in the pointers tree, the pointers of the ancestors of $j$ until we find an ancestor pointed
   to by an \direct\ pointer (and not a \recursive\ pointer).   See Figure~\ref{alg:find-parent}. In addition,
   we initialize the process of finding neighbors of $j$ to its right (i.e., with a bigger serial number)
   by inserting into the heap of  $j$  the ``final node'' $n+1$ as well as
  the first child of $v$.
  \ShortVersion
  \item
  Observe that any {\em subsequent}  \nn($j$) {\em query}  is to return a {\em child} of $j$ in the BA-graph.
  The children $x$ of $j$ in the BA-graph have, in the pointers tree, a path of $u(\cdot)$ pointers
   starting at $x$ and ending at $j$  with all pointers, except the last
    one, being \recursive\ (the last being \direct).  The query
    has to report
    the children in increasing index number. To this end the {\tt heap} of $j$ is used;
     it stores  some of the children of $j$, {\em  not yet returned by
      a} $\nn(v)$ {\em query}. This
     heap is also updated so that $\nn(j)$  will continue to  return the next child according to the
      index order.
To do so, whenever a node, $r$,  is extracted from the heap, 
 the heap is updated to  include the following:
  \ShortVersionEnd
  \LongVersion
  \item For any {\em subsequent} \nn($j$) {\em query} for node $j$ we proceed as follows.
  Observe that any subsequent query is to return a {\em child} of $j$ in the {\color{black} constructed} BA-graph.
  The children of $j$ in the BA-graph are those nodes $x$ which have, in the pointers tree, a path of $u(\cdot)$ pointers
   starting at $x$ and ending at $j$ and with all pointers on that path, except the last
    one, being \recursive\ (the last one being \direct).  The query $\nn(j)$ has, however, to report
    the children in increasing order of their index. To this end  the {\tt heap} of node $j$ is used;
     it stores at any give time some of the children of $j$ in the BA-graph, {\em  not yet returned by
      a} $\nn(j)$ {\em query}. We further have to update this
     heap so that $\nn(j)$  will continue to return the next child according to the
      index order.
  To this end we proceed as follows. Whenever  node, $r$ is extracted from the heap, in order to be returned as the next child,
  we update the heap to  include the following:
\LongVersionEnd
  \begin{itemize}
 \item  If $r$  has an \direct\ pointer to $j$, then we add to the heap (1) the next node, after $r$, with
  an \direct\ pointer to $j$, and (2) the first node that has a \recursive\ pointer to $r$.
  \item   If $r$  has a \recursive\ pointer to a node $r'$, then we add to the heap (1) the first node, after $r$, with
   a \recursive\ pointer to $r'$, and (2) the first node that has a \recursive\ pointer to $r$.
  \end{itemize}
  \end{itemize}

\LongVersion
The proof of Lemma~\ref{le:heap_correct} below is based on the premise that
 the heap  contains only children of $v$ in the BA-graph, and that it always contains
   the child of $v$ just after the one last returned.
   \LongVersionEnd
\ShortVersion
The proof of the next lemma,  by induction on the number of queries,
 is omitted.
\ShortVersionEnd

\begin{lemma}
\label{le:heap_correct}
The procedure \nn\ returns the next neighbor of $v$.
\end{lemma}
\LongVersion
\begin{proof}

Given a pointers tree we define the following notions:
\begin{itemize}
\item The set of nodes which have an \direct\ pointer to a given node $j$. That is,
 for $1 \leq j \leq n$,  $D(j) \triangleq  \{i  ~|~ u(i)=j, flag(i)=\direct\}$.
 \item The set of nodes which have a \recursive\ pointer to a given node $j$. That is,
 for $1 \leq j \leq n$,  $R(j) \triangleq  \{i  ~|~ u(i)=j, flag(i)=\recursive\}$.
\end{itemize}
Given a BA graph, for any node $1 \leq j \leq n$ and any prefix length $0 \leq \ell \leq n-1$,
we denote by $N^{\ell}(j)$
the set of the first (according to the index number) $\ell$ neighbors of $j$ in the BA graph.

\medskip
In what follows we consider an arbitrary node $j$.
We  consider the actions of  \nn\ (see~Figure~\ref{alg:next_neighbor}).
Let $M^{\ell}(j)$ be the set of nodes returned by the first $\ell$ calls $\nn(j)$.
We first prove
 that the following invariant holds. \\
 Just after call number $\ell \geq 1$ of $\nn(v)$:
\begin{enumerate}
\item
\label{inv:only}
 The heap $heap_j$ contains only neighbors of $j$ in the BA-graph.
 \item
 \label{inv:direct}
 The heap $heap_j$ contains the minimum node in  $ D(j) \setminus M^{\ell}(j)$.
 \item
  \label{inv:recursive}
  Let $q$ be the first neighbor of $j$ in the BA graph.
 The heap $heap_j$ contains, for each node $i \in M^{\ell}(j) \setminus \{q\}$, the minimum node in $R(i) \setminus M^{\ell}(j)$.
\end{enumerate}

We prove that the invariant holds by  induction on $\ell$.
The induction basis, for call number $\ell=1$,  holds since  (1) the first call to $\nn(j)$ results in inserting into  $heap_j$
the first node $x$ which has an \direct\ pointer to node $j$ and (2)  $heap_j$ was previously empty (see~Figure~\ref{alg:next_neighbor}).
Thus all points of the invariant hold after call $\ell=1$.
For $\ell >1$ assume that the induction hypothesis holds for $\ell-1$ and let $r$ be the node returned by the $\ell$'th
call to $\nn(j)$. We claim that the invariant still holds after call $\ell$ by verifying each one of the two cases for
the pointer of $r$ and the insertions into the heap for each such case. \\
If $r$ has an \direct\ pointer, then the following nodes are inserted into $heap_j$:
(1) The first node after $r$ with an \direct\ pointer to $j$.  Since this is a neighbor of $j$  in the BA graph
Point~\ref{inv:only} continues to hold. Since
 $r$, just extracted from the heap, was the minimum node in the heap, Point~\ref{inv:direct} continues to hold.
(2)  The first node after $r$ which has a \recursive\ pointer to $r$. Since this is a neighbor of $j$ in the BA graph
Point~\ref{inv:only} continues to hold; Point~\ref{inv:recursive} continues to hold since nothing has changed for any
other $i\neq r$, $i \in M^{\ell}(j) \setminus \{q\}$, and for $r$ the minium node in $R(i) \setminus M^{\ell}(j)$ is just inserted. \\
If $r$ has a \recursive\ pointer, and let $q$ be the parent of $r$ in the pointers tree,  then the following nodes are inserted into $heap_j$:
(1) The  first node after $r$ which has a \recursive\ pointer to $q$; denote it $x$. Since $x$ is a neighbor of $j$ in the BA graph
Point~\ref{inv:only} continues to hold. Since
 $r$, just extracted from the heap, was the minimum node in the heap, $x$ is the minimum node in
$R(i) \setminus M^{\ell}(j)$ and Point~\ref{inv:recursive} continues to hold (nothing changes for any $q'\neq q$, $q'\in  M^{\ell}(j) \setminus \{q\}$).
(2) The first node after $r$ which has a \recursive\ pointer to $r$. The same arguments as those for the corresponding case when  $r$ has
an \direct\ pointer hold, and thus both Point~\ref{inv:only} and Point~\ref{inv:recursive} continue to hold.\\
This concludes the proof of the invariant.

We now use the above invariant in order to prove that, for any $\ell \geq 1$, $N^{\ell}(j)= M^{\ell}(j)$.
We do this by induction on $\ell$.
For $\ell=1$ the claim follows from the facts the first neighbor of node $j$ is its parent in the BA graph and that the first
call $\nn(j)$  returns the value that $\fparent(j)$ returns. This proves the induction basis.
We now prove the claim for $\ell >1$ given the induction hypothesis for all $\ell' < \ell$.
Let node  $x$  be the $\ell$'th neighbor of $j$.
We have two cases: (1) node $x$ has an \direct\ pointer to $j$; (2) node $x$ has a \recursive\ pointer to another child  of $j$ in
 the BA graph (i.e., to another neighbor of $j$ in the BA graph, which is not  the first neighbor).

\noindent  Case (1): By the induction hypothesis $N^{\ell-1}(j)= M^{\ell-1}(j)$, hence by Point~\ref{inv:direct} of the invariant
$x$ is in the heap $heap_j$ when the $\ell$'th call occurs. Since any node returned by $\nn(j)$ is no longer in $heap_j$, by
Point~\ref{inv:only} of the invariant,
$heap_j$ does not contain any node smaller than $x$. Therefore the node returned by the $\ell$'th call of $\nn(j)$ is node $x$.

\noindent Case (2): Let node $y$ be the parent of node $x$ in the pointers tree, i.e., $u(x)=y$. Since $y$ is a neighbor
of $j$ in the BA graph, and $y <x$, it follows that $ y \in N^{\ell-1}(j)$, and by the induction hypothesis $y \in M^{\ell-1}(j)$.
Moreover, any node $x' <x $  has  $u(x')=y$, $flag(x')=\recursive$ if and only  if it is a neighbor of $j$, hence
any such node $x'$ is in $N^{\ell-1}(j)$, and by the induction hypothesis also in $M^{\ell-1}(j)$. It follows from Point~\ref{inv:recursive}
of the invariant that $x$ is in the heap $heap_j$  when the $\ell$'th call occurs.  Since any node returned by $\nn(j)$ is no longer
in $heap_j$, by Point~\ref{inv:only} of the invariant,
$heap_j$ does not contain any node smaller than $i$. Therefore the node returned by the $\ell$'th call of $\nn(j)$ is node $i$.
This completes the proof of the lemma.
\end{proof}
\LongVersionEnd

\begin{lemma}
For any given root (non-recursive) call of \fparent, with high probability,  that call takes  $O(\log^2 n)$ time.
\end{lemma}

{\color{black}
\begin{proof}
Consider the execution of \fparent\ as stated in Figure~\ref{alg:find-parent}.
Consider the recursive invocation tree that results from a call to \fparent.
Each path in this tree corresponds to a path in the pointers tree.
Since, By Claim~\ref{cl:low_degree_and_height}, w.h.p. the height of the pointers tree is bounded by $O(\log n)$, the lemma follows by Lemma~\ref{le:complexity_pointers_tree}.
\end{proof}
}

\LongVersion
We can now conclude with the following theorem.
\LongVersionEnd
\ShortVersion
The next theorem follows from the code,  standard heap implementation, and
 Lemma~\ref{le:complexity_pointers_tree}.
\ShortVersionEnd
\begin{theorem}
\label{th:comp_one_call_nn}
For any given call of \nn, with high probability,  all of the following hold for that call:
\ShortVersion
\begin{inparaenum}[(1)]
\ShortVersionEnd
\LongVersion
\begin{enumerate}
\LongVersionEnd
\item The increase, during that call, of the space used  by our algorithm is $O(\log^3 n)$.
\item The number of  random bits used during that call is $O(\log^5 n)$.
\item The time complexity of that call is $O(\log^6 n)$.
\LongVersion
\end{enumerate}
\LongVersionEnd
\ShortVersion
\end{inparaenum}
\ShortVersionEnd
\end{theorem}
\LongVersion
\begin{proof}
{\color{black}
Each call of \nn\ is executed, w.h.p., by $O(\log n)$ number of calls (and a constant number of calls in expectation) to \fparent\ and  \nuct,
 as well as $O(\log n)$ number
 of calls to \hinsert\ and \hextmin, and accesses to ``arrays''.}
  The claim then follows from standard deterministic
 heap implementations (which uses $O(1)$ space per  stored item  and $O(\log n)$ time per query) and from
 Lemma~\ref{le:complexity_pointers_tree}.
\end{proof}
\LongVersionEnd

We  now state the properties of our on-the-fly graph generator for BA-graphs.

\begin{definition}
\label{def:answers}
For a number of queries $T>0$ and a sequence of  \nn\ queries $Q=(q(1),\ldots,q(T))$,
let $A(Q)$ be the sequence of answers returned by an algorithm $A$ on $Q$.
If $A$ is randomized then $A(Q)$ is a probability distribution on sequences of answers.
\end{definition}

Let $\optBA_n$ be the (randomized) algorithm that first runs the Markov process to generate a graph $G$ on $n$ nodes
according to the BA  
model,
 stores $G$, and  then answers
 queries by accessing  the stored $G$.
 Let  $\flyBA_n$ be the algorithm \flyBA\
 run with graph-size $n$.
\LongVersion
 From  the definition of the algorithm we have the following.
\LongVersionEnd

\begin{theorem}
\label{th:equal_distribution}
For any sequence of queries $Q$, $\optBA_n(Q) = \flyBA_n(Q)$.
\end{theorem}

\LongVersion
We  now conclude by stating the complexities of our on-the-fly BA generator.
\LongVersionEnd
\begin{theorem}
For any $T>0$ and any sequence of queries
 $Q=(q(1),\ldots,q(T))$, when using $\flyBA_n$ it holds
  w.h.p. that, for all $1\leq t \leq T$:
\ShortVersion
\begin{inparaenum}[(1)]
\ShortVersionEnd
\LongVersion
\begin{enumerate}
\LongVersionEnd
\item The increase in the used space, while  processing query $t$,  is $O(\log^3 n)$. {\color{black} Therefore, the total space that is used to store the state of the generator after $T$ queries is $O(T \cdot \log^3 n)$.}
\item The number of  random bits used  while processing  query $t$ is $O(\log^5 n)$.
\item The time complexity for processing query $t$ is $O(\log^6 n)$.
\LongVersion
\end{enumerate}
\LongVersionEnd
\ShortVersion
\end{inparaenum}
\ShortVersionEnd
\end{theorem}
\begin{proof}
A query $\nn(v)$ at time $t$ is a {\em trivial} if  at some  $t'<t$  a query $\nn(v)$  returns
$n+1$.
Observe
 that trivial queries take $O(\log n)$ {\em deterministic} time, do not use  randomness, and do not increase the used space.
Since  there are less than $n^2$ {\em non-trivial}  queries, the theorem  follows from Theorem~\ref{th:comp_one_call_nn}
and a union bound.
\end{proof}

\LongVersion
We note that  the various assertions in this paper of the form of ``with high probability ... is $O(\log^cn)$'' can also
be stated in the form of   ``with probability $1-\frac{1}{n^d}$ ... is $f(d)\cdot \log^c n$''.
Therefore, we can combine these
 various assertions, and together with the fact that the number of non-trivial queries
 is $\poly(n)$, we  get the final result stated above.
 \LongVersionEnd

\LongVersion

\LongVersionEnd

{\color{black}
\section{A Direction for Extending to General Out-Degrees}
Given a random process that generates a Preferential-Attachment graph with out-degree $1$, Bollob{\'{a}}s and Riordan~\cite{BR04} consider the following generalization for defining a process that generates a Preferential-Attachment graph with general out-degree.
A $BA^{m}_{n}$-graph, where $BA^{m}_{n}$ denotes an $n$ node Preferential-Attachment graph with out-degree $m$, is constructed from a $BA_{nm}$-graph by identifying each of the $n$ nodes of the $BA^{m}_{n}$-graph with a block of $m$ nodes in the $BA_{nm}$-graph~\footnote{\color{black}Bollob{\'{a}}s and Riordan define the process $BA_{n}$ slightly different from the model we are using in this paper. According to their definition the head of the edge $e_j$ is the node $v_i$, for $1\leq i \leq j-1$,
with probability  $\frac{\deg(v_i,BA_{j-1})}{2j-1}$ and is $v_j$ with probability $\frac{1}{2j-1}$ (so there is some, vanishing, probability of forming self-loops).}

Hence, a $BA^{m}_{n}$-graph can be generated on-the-fly by using an on-the-fly generator for a $BA_{n m}$-graph, in a black-box manner, by increasing the complexity by a factor of $m$ (up to poly-logarithmic factors).
We leave working out the details for this generalization to higher out-degrees, as well as for other variants of generalization, for further research.

}
\LongVersion
\section{Conclusions}
We introduce an approach to probing and accessing a huge random graph, sampled from a given distribution,
in a way that does not require
to first sample the whole huge graph and then access it. The latter approach may require prohibitive amounts of time,
space and randomness which could be avoided if only relatively small parts of the random graph are accessed. The
feasibility of such savings may be especially challenging when the distribution at hand is usually defined by an evolving
sequential process (over, e.g., nodes) such as in the Barab{\'{a}}si-Albert Preferential Attachment
model or the random recursive tree model.

We show how to achieve such savings for the Barab{\'{a}}si-Albert graphs of out-degree $1$, as well as for the
evolving tree model, and give on-the-fly generation algorithms for both models such that with
probability  $1-1/\poly(n)$, each and every query is answered in $\polylog(n)$ time, and the
increase in space and the number of random bits consumed by any single query are both
$\polylog(n)$, where $n$ denotes the number of vertices in the graph.
\LongVersionEnd

\LongVersion
\paragraph*{Acknowledgments.}
\LongVersionEnd
\ShortVersion
\paragraph{Acknowledgments.}
\ShortVersionEnd
We thank Yishay Mansour for raising the question of whether one can locally generate
preferential attachment graphs, and  Dimitri Achlioptas and Matya Katz for useful
discussions. We further thank an anonymous ICALP reviewer for a comment that helped us
simplify one of the data structure implementations.

\bibliographystyle{plain}
\bibliography{pa}

\LongVersion
\appendix

\section{Implementations of Data Structures}

\subsection{Data Structure for $\varphi(\cdot)$}
\label{se:stabbing_data_structure}

We use two balanced binary search trees (or order statistic trees).
 One, called {\tt left}, stores all vertices $i$ such
that $\front(i)\neq \nil$. The other, called {\tt right}, stores (the multi-set) $\{\front(i) ~| ~\front(i)\neq \nil\}$.
To determine $\varphi(a)$ we find, using tree {\tt right}, how many nodes $i$ have $\front(i) > a-1$ (and $\front(i)\neq \nil$).
Let this number be $R$. Using tree {\tt left} we find how many nodes  $i <a$ have  $\front(i)\neq \nil$.
Let this number be $L$. Then $\varphi(a)= R-L$.

By standard implementations of balanced search trees
 the space complexity is $O(k)$ and all operations are done in time $O(\log k)=O(\log n)$) Here $k$ denotes the number of nodes $i$ such that  $\front(i)\neq \nil$.

\subsection{Data Structure to find the node of rank $h$ in $[a,n+1]\setminus K$}
\label{se:rank_data_structure}

We start with a number of definitions useful for specifying the data structure and its operations.

For a node $j\in \{1, \ldots , n+1\}$ and a subset of nodes $Q \subseteq \{1, \ldots , n+1\}$,  define $Q(j)$ as follows:
$$
{Q}(j) \triangleq
\left\{
\begin{array}{ll}
j  &\text{if} ~ j \in Q  ~\text{or} ~ j=1 \\
\max_{j' \in Q} \{j' | j'<j\} & \text{otherwise\,} 
\end{array}
\right.~.
$$

Note that for technical reasons for $j=1$ we define $Q(j)=1$ whether or not $j \in Q$.

For a node $j\in \{1, \ldots , n+1\}$ and a subset of nodes $Q \subseteq \{1, \ldots , n+1\}$, define $\trank_{Q}(j)$ as follows:
$$
\trank_{Q}(j) \triangleq  | \{i ~| ~i < Q(j) ; i \in Q\} |~.
$$

We note that using these definitions we have that, for any $j \in \{1, \ldots , n+1\}$, the number
 of items $ i <j $ in $\bar{Q}$, where $\bar{Q} = \{1, \ldots , n+1\} \setminus Q$, is  $(j-1) - \trank_Q(j)$.

 The $\tinsert_Q$, $\tdelete_Q$ and $\trank_Q$ operations are implemented as in a standard order-statistics tree
 based on a balanced binary  search tree.
 The operation $\trank_{\bar{Q}}$ is implemented using the $\trank_Q$  and then performing the calculation above.
 To implement  $\tselect_{\bar{Q}}(s)$ we proceed as follows.
 We traverse the search tree with the value $s$, and in each node of the tree that contains the vertex $j$
  we compare $s$ with $(j-1)-\trank_{Q}(j)$. Thus, we can find the maximum $j\in Q$
  such that $\trank_{\bar{Q}}(j) \leq s$. Denote this node $j'$. We then return the
 node $j' +\left[ s +1 - ((j'-1)-\trank_{Q}(j'))\right]$.

 The time complexities of $\tinsert_Q$, $\tdelete_Q$ and $\trank_Q$ and $\trank_{\bar{Q}}$ are therefore $O(\log n)$ based on
  standard order statistics trees. The time complexity of $\tselect_{\bar{Q}}$ is $O(\log^2 n)$: for each node along the search path of
  length $O(\log n)$ we need to use the query $\trank_Q$.
\LongVersionEnd

\end{document}